\documentclass[final]{econsocart}

\usepackage{booktabs}
\usepackage[yyyymmdd,24hr]{datetime}
\usepackage{caption}
\usepackage{setspace}
\usepackage{subcaption}
\usepackage{mathtools}
\RequirePackage[colorlinks,citecolor=cyan,linkcolor=blue,urlcolor=blue,pagebackref]{hyperref}

\startlocaldefs
\usepackage{enumerate}

\newcommand{\new}[1]{\textcolor{black}{#1}}

\theoremstyle{plain}

\newtheorem{theorem}{Theorem}
\newtheorem{proposition}{Proposition}

\newtheorem{lemma}{Lemma}
\newtheorem{assumption}{Assumption}

\newtheorem{corollary}{Corollary}
\theoremstyle{remark}
\newtheorem{definition}{Definition}

\usepackage{tikz}
\usetikzlibrary{shapes}

\usepackage{etoolbox}

\newcommand{\da}{\mathtt{DA}} 
\newcommand{\eada}{\mathtt{M}}

\newcommand{\rk}{\textup{rk}}

\theoremstyle{remark}

\AtBeginEnvironment{definition}{\vspace{-2.5ex}}

\renewcommand{\appendixname}{Supplementary Appendix}
 
\usepackage{enumitem}
\usepackage{amstext} 
\usepackage{array}   
\newcolumntype{C}{>{$}c<{$}} 
\usepackage{multirow}
\usepackage{float}
\usepackage{array}   
\newcolumntype{C}{>{$}c<{$}} 
\endlocaldefs
\begin{document}
	\begin{frontmatter}
		
		\title{The Large and Likely Inefficiency of Stable Matching Mechanisms}
		\runtitle{The Inefficiency of Stable Matchings}
		
		\begin{aug}
			\singlespacing 
			\author[id=au1,addressref={add1}]{\fnms{Josu\'e}~\snm{Ortega}}
			\author[id=au2,addressref={add2}]{\fnms{Gabriel}~\snm{Ziegler}}
			\author[id=au3,addressref={add3}]{\fnms{R. Pablo}~\snm{Arribillaga}}
			\author[id=au4,addressref={add4}]{\fnms{Geng}~\snm{Zhao}}
			
			\address[id=add1]{\orgname{Queen's University Belfast}}
			
			\address[id=add2]{\orgname{Freie Universität Berlin and Berlin School of Economics}}
			
			\address[id=add3]{\orgdiv{Instituto de Matemática Aplicada San Luis},
				\orgname{Universidad Nacional de San Luis, and CONICET}}
			
			\address[id=add4]{\orgname{University of California, Berkeley}}
		\end{aug}
		
		\support{Emails: j.ortega@qub.ac.uk, gabriel.ziegler@fu-berlin.de, rarribi@unsl.edu.ar, gengzhao@berkeley.edu\\
			Manuscript resubmitted to Theoretical Economics. \hfill{\it {\ddmmyyyydate\today}}
		}
		
		
		
		\begin{abstract}
			\begin{center}
				\bf Abstract
			\end{center}
We prove that any stable matching mechanism suffers from systematic inefficiency of striking magnitude: in large random markets, any stable allocation is Pareto-inefficient with high probability, and almost all students can simultaneously improve their placements without harming anyone else.
We establish this result by showing that the envy digraph generated by the student-proposing Deferred Acceptance mechanism contains a unique giant strongly connected component, implying that nearly all students are improvable via trading cycles.
Finally, we show that every maximal cycle packing covers almost all students, revealing a surprising asymptotic equivalence among all efficient mechanisms that Pareto-dominate DA.
\end{abstract}


	\end{frontmatter}
	\newpage
	\section{Introduction}
	\label{sec:introduction}
	\setlength{\parskip}{0.25em} 

Stability is a central principle in the design of school choice mechanisms. It eliminates justified envy, promotes perceived fairness, and underlies the widespread adoption of the student-proposing Deferred Acceptance (DA) algorithm, which generates the student-optimal stable matching. Thus, if stability is a fundamental design desideratum, DA provides the natural benchmark.

Yet stability comes at a cost, because even the student-optimal stable matching need not be Pareto efficient.
Some students could receive strictly better assignments without making anyone worse off, and the degree of this inefficiency may be significant both in theory and in practice.\footnote{DA may assign every student to one of their least-preferred schools \citep{kesten2010school}, and approximately 2\% of students in the New York City high school match, which uses a version of DA, could have been assigned to a more preferred school without harming others \citep{abdulkadirouglu2009strategy}.} 
This observation raises two important questions: \emph{how many students could improve their DA placement without harming others}, and \emph{how many of them could do so simultaneously}? These are the questions we address in this paper. 

\new{To answer these questions, we focus on \emph{improvable} students: those who could be assigned to a strictly more preferred school under some Pareto improvement upon the DA matching (or under any other stable matching, by DA's student-optimality). It is well-known that improvable students exactly correspond to cycles in DA's directed envy graph, where students point to all those whose placements they envy. Consequently, they belong to non-trivial strongly connected components of this graph (SCCs), that is, sets of more than one node where everyone can reach each other by envy relations.}

\new{The SCC representation of improvable students allows us to answer our first question: how many students are improvable? In random matching markets with independently and uniformly distributed preferences, we prove that DA's envy graph admits a unique giant SCC with high probability}, implying that the fraction of unimprovable students converges to zero as market size grows (Theorem \ref{thm:scc}). The underlying reason is that DA generates a very dense envy digraph: most students make a logarithmic number of applications, so the probability that two sizable groups of students never envy each other becomes vanishingly small. As a consequence, the envy digraph develops a familiar structure in homogeneous random graphs: a unique giant SCC in which almost all students lie. In large markets, therefore, almost every student can be strictly better off relative to any stable matching without harming anyone.
	
Although DA's Pareto-inefficiency was previously known, Theorem \ref{thm:scc} reveals that its pervasiveness is not confined to pathological examples or specific datasets but emerges as a systematic property in random markets. The consequences are severe. We show that the vast majority of students could improve their placement without harming others. Our result implies that, while the most disadvantaged students cannot benefit from enhanced efficiency (i.e. those assigned to schools near the bottom of their preference lists, or not assigned at all), the vast majority are improvable, revealing both the substantial scope for welfare improvements and the fundamental inefficiency of stability in large markets.
	
We next ask how much of this potential can be realized simultaneously. One might object that these potential improvements are mutually exclusive, so that helping one student precludes helping another. We prove otherwise. Theorem \ref{thm:cyclepacking} establishes that any mechanism Pareto-dominating DA improves essentially the same set of students---almost everyone---in large markets.

This surprising asymptotic equivalence has important implications. While \citet{knipe2025improvable} show that in carefully constructed instances, different efficient mechanisms can vary dramatically in the number of students they improve, these differences vanish asymptotically.\footnote{EADA and DA+TTC may not only improve fewer students than alternative mechanisms that Pareto-dominate DA in specific instances, but may do so while generating more justified envy \citep{knipe2025improvable}.} Mechanisms such as \citeauthor{kesten2010school}'s Efficiency-Adjusted Deferred Acceptance (EADA) or Top Trading Cycles using DA's allocation as an endowment become virtually indistinguishable in terms of coverage as markets grow. The high density of DA's envy digraph ensures that if a student is excluded from one trading cycle, that student is likely to participate in another.
Once one departs from stability while requiring Pareto-dominance over DA, design choices should therefore be guided by criteria other than the number of improved students---such as fairness interpretations or susceptibility to manipulation, both areas in which EADA has shown promise.
	
Throughout the paper, we state and prove our main results for one-to-one matching markets with independently and uniformly distributed preferences and priorities. These assumptions are imposed for clarity and tractability. However, our results extend beyond this baseline setting. In Appendix \ref{app:appendixone}, we prove that both the unique giant SCC and the asymptotic equivalence of efficient mechanisms extend to many-to-one matching markets where schools have multiple but constant seats, though convergence rates depend on school capacity (Theorems \ref{thm:many-to-one-scc} and \ref{thm:cyclepacking-many-to-one}). In Appendix \ref{app:appendixtwo}, we establish robustness to preference correlation through tiered preferences (Theorem~\ref{thm:tiered}) and general correlated preferences satisfying sufficient regularity conditions, including bounded rank-ordered logit preferences (Theorem~\ref{thm:correlated-giant} and Corollary~\ref{cor:logit-giant}). 
\new{ In Appendix \ref{app:appendixthree}, we use simulations to examine the role of correlation in priorities to find that our results are robust to mild and moderate (but not high) correlation. When we introduce correlation in both sides, the unique giant SCC becomes even larger than under the iid baseline. Thus, the pervasiveness of improvable students under DA, and hence under any other stable matching mechanism, remains a robust phenomenon across a wide range of realistic market structures.}

	\section{Related Literature}
	\label{sec:literature}

Our work contributes to several strands of literature on matching mechanisms and random market models. We organize our discussion around four interconnected themes.

\paragraph{Measuring DA's inefficiency.}
The Pareto-inefficiency of DA has long been recognized \citep{abdulkadirouglu2003, kesten2010school}. Recent work quantifies this inefficiency using different welfare criteria. \citet{lee2014efficiency} and \citet{che2019efficiency} study \emph{utilitarian} welfare and show that, under particular assumptions in large markets, DA can achieve near-optimal average utility. This is not in tension with our results. Pareto inefficiency is an \emph{ordinal} notion: it can be widespread even when aggregate utility is high, for example when schools are close substitutes so that the utility gaps between adjacent ranks are small. Our contribution focuses on the prevalence of Pareto improvements over DA within the class of stable matchings and does not require cardinal interpersonal utility comparisons. Empirically, our objects of interest can be assessed from submitted preferences, whereas utilitarian evaluations require estimating a cardinal preference model.

\cite{escobar2025efficiency} study the top-choice share---the fraction of students receiving their first choice---in a continuum model with specific priority structures. Since students getting their top choice are unimprovable, our Theorem~\ref{thm:scc} implies this fraction becomes negligible in large markets\new{, a result also established in more generality by \cite{ashlagi2019assigning} for DA with random and independent priorities, which they show no longer holds when there is a single master priority list used by all schools.} \footnote{Beyond DA, \cite{pritchard2023asymptotic} and \cite{ortegascwe} analyze top-choice shares for the Boston mechanism.}

Other studies use measures related to the price of anarchy. \cite{anshelevich2013anarchy} examine the ratio between utilitarian welfare under efficient versus stable matchings, while \cite{boudreau2013preferences} develop an ordinal variant using rank sums. Both metrics require interpersonal welfare comparisons and are silent about the distribution of welfare gains. 

 Necessary and sufficient conditions on priorities are known for DA to be inefficient for some preference profile \citep{ergin2002efficient}. Empirical studies have documented DA's inefficiency in practice \citep{abdulkadiroglu2005, abdulkadirouglu2009strategy, che2019efficiency, ortega2023cost}.

\paragraph{Efficient mechanisms that dominate DA.} The most prominent mechanism in this class is Efficiency-Adjusted Deferred Acceptance \citep[EADA,][]{kesten2010school}. Over the past decade, EADA's properties and implementation have been extensively studied \citep{bando2014existence, tang2014new, dur2019school, troyan2020essentially, troyan2020obvious, ehlers2020legal, tang2021weak, dougan2021minimally, reny2022efficient, chen2023regret,dogan2023existence,afacan2025improving,shirakawasimple}, demonstrating that it is possible to achieve an efficient improvement over DA while maintaining relatively low instability and manipulability. Another well-known efficient mechanism that Pareto-dominates DA is DA+TTC, which applies the top trading cycles (TTC) procedure to the allocation obtained by DA \citep{alcalde2017}. Beyond EADA and DA+TTC, our companion paper shows that any mechanism that weakly Pareto-dominates DA can lead to a poor rank distribution and high levels of student segregation \citep{ortega2025pareto}.

Two seminal papers have studied the family of mechanisms that Pareto dominate DA: \cite{alva2019stable} provide a comprehensive study of stable-dominating rules, which include any efficient mechanism that improves on student-proposing DA, and \cite{tang2014new}, who introduce the concept of improvable students and under-demanded schools. While these foundational works identify the existence of improvable students, they do not quantify their prevalence---gaps we address by deriving asymptotic bounds.\footnote{\citet[][p. 548]{tang2014new} conjectured that the set of unimprovable students is large. Our work demonstrates that this is not the case in large markets.} 
	
\paragraph{Empirical evaluations of mechanisms that dominate DA.} \cite{diebold2017matching} conduct an extensive comparison across course allocation datasets, finding that the difference between DA and EADA rank distributions is not statistically significant in 7 out of 19 datasets. They also show that EADA and DA always match the same number of students. \cite{ortega2023cost} compute counterfactual EADA allocations for Budapest's school choice system and find that, while EADA improves the average student placement, it does not help unassigned students or the worst-placed pupil (assigned to his 13th ranked school under both mechanisms) or those unassigned. 
In laboratory experiments, \cite{cerrone2022school} find that EADA generates Pareto-efficient allocations more frequently than DA in the lab but note that the distribution of efficiency improvements is uneven among participants. \new{\cite{freer2026experimental} find that the gains become imprecisely estimated in alternative lab experiments with higher manipulation rates}.
	
\paragraph{Random matching markets.} Our analysis of the proportion of unimprovable students builds on a rich literature on random matching markets at the intersection of economics and computer science \citep{wilson1972, knuth1976, pittel1989, pittel1992likely, immorlica2005, kojima2009, che2010asymptotic, lee2014efficiency, lee2016incentive, liu2016, ashlagi2017, che2018payoff, che2019efficiency, pycia2019evaluating, ashlagi2020tiered, ashlagi2023welfare, nikzad2022rank, ortega2023cost, ronen2025stable}. Our paper advances this literature by providing a characterization and quantification of unimprovable students, introducing new analytical tools that connect matching theory to random graph theory.
The emergence of giant strongly connected components in random digraphs is a well-studied phenomenon in graph theory, but its application to understanding efficiency in matching markets represents a novel theoretical contribution.

Finally, our Theorem \ref{thm:cyclepacking} is related to a number of asymptotic results in matching markets. While previous equivalence theorems focus on anonymous statistics such as normalized rank distributions or average utilitarian efficiency \citep{che2010asymptotic,  liu2016, pycia2019evaluating, che2018payoff}, we prove equivalence in terms of the number of students who improve over DA. Since this statistic is non-anonymous and mechanisms that Pareto dominate DA are not strategy-proof, existing equivalence results---in particular Pycia's comprehensive framework---do not apply to our setting, making Theorem \ref{thm:cyclepacking} a distinct theoretical contribution.

\section{Model}
\label{sec:model}

Following \cite{abdulkadirouglu2003}, a school choice problem $P$ consists of a set of students $I$ and a set of schools $S$. For simplicity, we shall assume that $|I| = |S| = n$, although Appendix \ref{app:appendixone} extends our results to many-to-one problems. Each student $i$ has a strict preference $\succ_i$ over the schools. Each school $s$ has one seat and a strict priority $\triangleright_s$ over the students, determined by local educational regulations.

For a given school choice problem $P$, a \textit{matching} $\mu$ is a bijection from $I$ to $S$. We denote by $\mu_i$ the school to which student $i$ is assigned and by $\mu^{-1}_s$ the set of students assigned to school $s$. 

The function $\rk_i:S \rightarrow \{1, \ldots, n\}$ specifies the rank of school $s$ according to the preference profile $\succ_i$ of student $i$:
\begin{align*}
	\rk_i(s) = |\{s' \in S: s' \succ_i s\}|+1,
\end{align*}
so that the most desirable option gets a rank of 1, whereas the least desirable gets a rank of $n$. With some abuse of notation, we use the same rank function to specify the students' rank per the priority profile of schools, i.e. $\rk_s(i)= |\{j \in I: i \triangleright_s j\}|+1$.

A matching $\mu$ \textit{weakly Pareto-dominates} matching $\nu$ if, for every student $i \in I$, $\rk_i(\mu_i) \leq \rk_i (\nu_i)$. A matching $\mu$ \textit{Pareto-dominates} matching $\nu$ if $\mu$ weakly dominates $\nu$ and there exists a student $j \in I$ with $\rk_j(\mu_j) < \rk_j (\nu_j)$. A matching is \textit{Pareto-dominated} if there exists a matching that Pareto-dominates it and is \textit{Pareto-efficient} if it is not Pareto-dominated. 

Student $i$ \emph{desires} school $s$ in matching $\mu$ if $\rk_i(s)< \rk_i(\mu_i)$ and he \emph{envies} student $j$ at matching $\mu$ if $\rk_i(\mu_j)<\rk_i(\mu_i)$. We say that student $j$ \emph{violates} student $i$'s priority at school $s$ in matching $\mu$ if $i$ envies $j$, $\mu_j =s$, and $\rk_s(i) < \rk_s(j)$. 
A matching $\mu$ is \emph{non-wasteful} if every school that is desired by some student is matched to a student. 
A matching $\mu$ is \textit{stable} if it is non-wasteful and no student's priority at any school is violated in $\mu$.

A \emph{mechanism} associates a matching to every school choice problem. We are mainly interested in the \emph{student-proposing Deferred Acceptance (DA) mechanism} \citep{gale1962}, which works as follows:

\begin{enumerate}[leftmargin=2cm]
	\item[Round 1:] Every student applies to her most preferred school. Every school tentatively accepts its most preferred applicant according to its priority and rejects the rest. 
	\item[Round $k$:] Every student rejected in the previous round applies to her next best school. Among new applicants and any previously accepted student, every school tentatively accepts its most preferred student according to its priority and rejects the rest.
\end{enumerate}

The procedure stops when there is a round without any new rejection. We use $\da(P)$ to denote the unique resulting matching generated by DA in school choice problem $P$.  $\da_i(P)$ denotes the school to which student $i$ is assigned. $\da^{-1}_s(P)$ denotes the student assigned to school $s$ in $\da(P)$.

Similarly, we will use $\eada$ to denote an arbitrary mechanism that returns a Pareto-efficient allocation that weakly Pareto-dominates DA. We use $\mathcal M(P)$ to denote the class of such mechanisms for school choice problem $P$, which include Efficiency-Adjusted Deferred Acceptance (EADA) and Top Trading Cycles using DA's allocation as endowments (DA+TTC). 

We now introduce a key concept in this paper: unimprovable students.

\begin{definition}
	A student $i \in I$ is \textit{unimprovable} if, for every matching $\eada(P) \in \mathcal M(P)$, we have $\eada_i(P)=\da_i(P)$.
\end{definition}

Unimprovable students do not necessarily need to be in an adverse situation. For example, every student who is assigned to their most preferred school in DA is unimprovable. We denote by $U(P)$ the set of unimprovable students in school choice problem $P$. It is well-known that, for any school choice problem, there is always at least one unimprovable student \citep{tang2014new}.

Throughout the paper, we analyze the one-to-one matching case described above. In Appendix \ref{app:appendixone}, we extend our main results (Theorems \ref{thm:scc} and \ref{thm:cyclepacking}) to many-to-one matching markets where schools have multiple but constant quotas.
	
	
	\subsection{Students who Are Always Unimprovable}
	\label{sec:identifying}

	Recall that a digraph $G$ is given by a set of nodes $V$ and a set of directed edges $E$. We define DA's envy graph as follows.\\
	
	\begin{definition}
		Given a school choice problem $P$, we define DA's \emph{envy digraph} $G^{\da(P)}=(I, E(\da(P)))$ where each node corresponds to a student and edge $(i,j)$ exists if $i$ envies $j$ in DA.
	\end{definition}
	
Envy graphs have been studied in the literature, for instance by \citet{lipton2004approximately} in the context of indivisible-good allocation, and more closely by \citet{dur2019school}, who refer to the same object as a ``directed application graph''. \footnote{In \cite{dur2019school}, not all priority violations are admissible, so only a subset of Pareto-improving cycles is considered. When all priority violations are allowed, their directed application graph coincides with the envy digraph studied here.}

	 A \emph{cycle} of a digraph is a sequence of nodes $i_0,i_1,\ldots,i_j,i_0$ such that there is an edge between any consecutive nodes and no edge is repeated. A \emph{trading cycle} is a cycle in which every node appears only once, except for $i_0$. DA's envy digraph gives us a full characterization of unimprovable students by identifying cycles in $G^{\da(P)}$.
	
	\begin{proposition}[\cite{dur2019school}]\label{thm:carachterization}
		A student is unimprovable if and only if he does not belong to any cycle in $G^{\da(P)}$.
	\end{proposition}

	The characterization in Proposition \ref{thm:carachterization}, \new{which applies to every non-wasteful mechanism as shown by \cite{erdil2014strategy}}, will be key for us to quantify unimprovable students in a probabilistic framework in the next section using a graph-theoretical approach. 
Using this characterization, we now identify a subset of students who are always unimprovable
	
	\begin{proposition}
		\label{thm:twotypes}
		For any school choice problem $P$, if there is a student $i$ such that $\rk_i[\da_i(P)]=n$ 
		then $i$ is unimprovable.
	\end{proposition}
	
	\begin{proof}
If a student $i$ is assigned to his least preferred school, then $\da_i(P)$ must be the only under-demanded school because every other institution rejected student $i$, and since any student assigned to an under-demanded school is unimprovable \citep{tang2014new}, the conclusion follows. Note that there is at least one under-demanded school, e.g. the last one to receive an application in DA.\footnote{In a separate note, we show that in random i.i.d.\ markets the expected number of such schools
		(equivalently, of students whom nobody envies) equals $H_n$, the $n$-th harmonic number \citep{orteganote}.}

	\end{proof}
	
An analogous result, in a more general model by \cite{alva2019stable}, shows that every student who is unmatched under DA is also unimprovable. Together, both results show that students receiving the worst outcomes---those unassigned or matched to their least preferred school---are systematically excluded from potential improvements. 
This helps us to understand the empirical regularity documented in Section \ref{sec:literature}: efficiency adjustments over DA consistently fail to help the students who most need better placements.
However, these results leave an important question unanswered: how prevalent are unimprovable students in practice? If most students fall into the categories identified in Proposition \ref{thm:twotypes}, then efficiency improvements over DA would benefit only a small minority. Conversely, if most students are improvable, then the scope for welfare gains through alternative mechanisms could be substantial. To address this quantitative question, we turn to a probabilistic analysis of random school choice markets.

\subsection{Improvable Students as Strongly Connected Components}
\label{sec:components}
The characterization in Proposition~\ref{thm:carachterization} shows that improvable students are exactly those who belong to some directed cycle.
While this description is useful locally, our goal is to understand the \emph{global} structure of potential improvements: how many students can benefit, and how these improvements interact.
To move from local cycles to global structure, it is convenient to work with strongly connected components, defined as follows.

A directed graph is \emph{strongly connected} if for every pair of nodes there is a directed path from the first to the second and from the second to the first. 
A \emph{strongly connected component (SCC)} of a digraph $G=(I,E)$ is a maximal subgraph that is strongly connected. 
An SCC is \emph{trivial} if it consists of a single node and \emph{non-trivial} otherwise. 
The following Corollary follows trivially from Proposition \ref{thm:carachterization}.\footnote{Since Proposition \ref{thm:carachterization} is proved by \cite{dur2019school} for many-to-one problems, the non-trivial SCC characterization of improvable students immediately extends to many-to-one matching markets as well.}

\begin{corollary}
	\label{cor:scc}
	A student is improvable if and only if she belongs to a non-trivial	strongly connected component of $G^{\da(P)}$.
\end{corollary}

Note that if a student belongs to a non-trivial SCC, it must be part of a cycle with other nodes in the SCC and therefore is improvable.  Conversely, if a student is improvable, they must belong to a cycle and every node in this cycle can reach the others and be reached by them, thereby forming (part of) a non-trivial SCC, which the student we started out with belongs to.
Corollary~\ref{cor:scc} shifts the perspective from individual trading cycles to the component structure of DA’s envy digraph. 
It allows us to evaluate inefficiency by studying the size of strongly connected components: if most students lie in large components, then most are improvable. We now use this observation to address our main question.

\section{Quantifying (Un)Improvable Students}
	
Our main goal is to determine the expected fraction of improvable students in large markets by analyzing the non-trivial SCCs of DA's envy digraph.
To do so, we consider a \emph{random school choice problem} where strict preferences and priorities are drawn independently and uniformly at random.
We will use this framework assuming (as the majority of the random market literature) that the number of students and schools is equal and given by $n$. The benefit of imposing this strong assumption is that it allows us to simplify the analysis considerably.
We denote by $P_n$ a random instance of a school choice problem of size $n$ and by $G^{\da(P_n)}$ the corresponding random envy digraph generated by DA in such problem.\footnote{We discuss extensions to many-to-one matching and correlated preferences in Appendices \ref{app:appendixone} and \ref{app:appendixtwo}.} 

\new{As we study the behavior of markets as $n$ grows large, we use throughout the following standard asymptotic notation. 
For functions $f,g:\mathbb{N}\to\mathbb{R}_+$, we write $f(n)=O(g(n))$ if there exists a constant $C>0$ such that $f(n)\le C g(n)$ for all sufficiently large $n$; $f(n)=\Omega(g(n))$ if $g(n)=O(f(n))$; $f(n)=\Theta(g(n))$ if both $f(n)=O(g(n))$ and $f(n)=\Omega(g(n))$; $f(n)=\omega(g(n))$ if $f(n)/g(n)\to\infty$ as $n\to\infty$; and $f(n)=o(g(n))$ if $f(n)/g(n)\to 0$ as $n\to\infty$. 
Furthermore, we say that an event occurs with high probability if its probability tends to 1 as $n\to\infty$.}

While the geometric out-degree distribution of $G^{\da(P_n)}$ differs from the Poisson distribution characteristic of the standard Erd\H{o}s-R\'enyi random digraph model, we demonstrate below that a unique giant SCC also emerges in DA's random envy digraph with high probability.
	
	\begin{theorem}
		\label{thm:scc}
		Fix an arbitrary $\epsilon > 0$. With high probability, $G^{\da(P_n)}$ has a unique giant SCC containing at least $(1 - \epsilon)n$ students.
	\end{theorem}
	
	\subsection{Proof of Theorem \ref{thm:scc}}
	
	The proof strategy hinges on establishing that DA's envy random digraph is sufficiently dense to preclude fragmentation into two large disjoint components lacking directed edges between them. 
	Such a partition would necessitate that no student in the first component had applied to any school matched to students in the second component during the execution of DA---an event whose probability decreases super-polynomially as $n$ approaches infinity.
	We formalize this intuition through a sequence of lemmas. First, we demonstrate that with high probability, any substantial subset of students generates a significant number of applications during the allocation process. 
	Subsequently, we establish that the probability of absence of cross-component applications between large partitions is vanishingly small. 
	The proof concludes by applying the union bound across all possible partitions, establishing the desired result.
	
	\paragraph{Bounding the Number of Applications.}
	We use the following simple consequence of standard random matching results: for every fixed finite $L$, the probability that a given
	student makes fewer than $L$ applications tends to zero. Hence, with high probability, all but an arbitrarily small fraction of
	students make at least $L$ applications. This fixed-$L$ bound is enough for the union bound below; $L=L(\epsilon)$ will be
	chosen sufficiently large.
	\begin{lemma}\label{lem:many-applications}
		Fix $\epsilon \in (0,1)$ and an integer $L \geq 1$.
		With high probability, all but at most $\epsilon n$ students each
		make at least $L$ applications in the DA process.
	\end{lemma}
	\begin{proof}
		Let $\delta_n=\Pr(i\text{ makes fewer than }L\text{ applications})$ for a fixed student $i$. The cited random matching results imply
		that, for every fixed $L$, $\delta_n\to0$ as $n\to\infty$.\footnote{See \citet{knuth1976,pittel1989}; see also \citet[Theorem
			~3.2]{ashlagi2023welfare} for a characterization of rank distributions in any stable outcome.}
		Let $X_i$ be the indicator that student $i$ makes fewer than $L$ applications. Then
		\[
		\mathbb E\left[\sum_i X_i\right]=n\delta_n.
		\]
		By Markov’s inequality,
		\[
		\Pr\left(\sum_i X_i>\epsilon n\right)\le \frac{\delta_n}{\epsilon}\to0.
		\]
		Thus, with high probability, all but at most $\epsilon n$ students make at least $L$ applications.
	\end{proof}
	
	\paragraph{Partition Argument and Probability of No Cross-Applications.}
	Suppose that $G^{\da(P_n)}$ has two disjoint subsets $I_0, I_1 \subseteq I$, each of size at least $\epsilon n$, such that there is no directed edge from $I_0$ to $I_1$ in $G^{\da(P_n)}$. 
	By definition, $(i \to j)$ means student $i$ prefers $\da_j(P_n)$ to $\da_i(P_n)$. 
	So if there is no edge from $I_0$ to $I_1$, it means no student in $I_0$ prefers the school matched to any student in $I_1$ over its own match. 
	Equivalently, no student in $I_0$ applied to any school $\da_j(P_n)$ with $j\in I_1$ in the DA algorithm; otherwise, that application would produce an envy edge $i \to j$ if that school ultimately ended up matched to $j$.
	
	Hence, an equivalent statement is: If $G^{\da(P_n)}$ has an SCC of size \emph{at most} $(1-2\epsilon)n$, we can split $I$ into two blocks
	$I_0, I_1$, each at least $\epsilon n$, so that no student in $I_0$ ever applied to any school in $S_1 = \{\da_i(P_n) : i \in I_1\}$. 
	We first show that for a \emph{fixed} partition $(I_0,S_1)$, the probability of no-application across blocks is very small.
	
	\begin{lemma}\label{lem:no-crossapplication}
		Let $I_0 \subseteq I$ and $S_1 \subseteq S$ be subsets of students and schools with $|I_0|,|S_1|\ge \epsilon n$. 
		Then the probability that the students in $I_0$ in combination make at least $L n$ applications yet no one applies to any school in $S_1$ is at most
		$(1-\epsilon)^{Ln}$.
	\end{lemma}
	
	\begin{proof}
		If applications were all mutually independent, the claim would be a simple consequence of standard binomial concentration. 
		In DA, however, the applications depend on the history of rejections and hence see slight correlation. 	
		Instead, we consider a version of deferred acceptance with redundant applications (what \cite{knuth1976} calls the \emph{amnesiac algorithm}) where we generate students' preferences in a deferred manner and direct each application-to-be-made to a uniformly random school independently. That is, we allow a student to re-apply to a school that has already rejected her (and hence will do so again). 
		To get her true preferences, we simply remove the duplicates from the sequence of applications she will ever make. 
		It is easy to see that this version of amnesiac DA can be coupled exactly to the conventional DA to yield the same outcome.
		
		In the amnesiac algorithm, each student makes weakly more applications.
		Thus, it suffices to consider the event that the students in $I_0$ in combination make at least $L n$ applications in the amnesiac algorithm yet no one applies to any school in $S_1$.
		
		To understand this probability, we use the following coupling trick to handle the preferences (and applications) of students in $I_0$:
		First, generate an (infinite) sequence of i.i.d.\ uniform samples $\sigma=(\sigma_1,\sigma_2,\ldots)$ from the set of schools $S$ independent of the preferences of all schools and all students outside $I_0$;
		Then, run DA with redundant applications, and whenever an application is to be sampled from a student $i\in I_0$, supply with the next element from $\sigma$. We can verify inductively that, at any stage, this coupling yields the same distribution (as in the amnesiac algorithm) for the next applications to make---simply uniform and independent of all previous events.
		Our target no-cross-application event in the lemma statement implies that (1) at least $Ln$ elements from $\sigma$ are read before DA terminates and (2) none of these elements belong to the set $S_1$; these two combined imply that $\sigma_j \notin S_1$ for all $j=1,\ldots, Ln$. By construction, the sequence $\sigma_1,\ldots,\sigma_{Ln}$ are i.i.d.\ samples independent of all other sources of randomness, and thus $\Pr(\sigma_j\notin S_1 \;\forall j\in[Ln]) = (\Pr(\sigma_1\notin S_1))^{Ln}$.
		Since $|S_1| \ge \epsilon n$, we have $\Pr(\sigma_1\notin S_1) \le 1-\epsilon$, and our desired upper bound of $(1-\epsilon)^{Ln}$ follows immediately.
	\end{proof}

	\paragraph{Union Bound over All Subsets}
	
	We now combine Lemma \ref{lem:no-crossapplication} with a union bound over all possible realizations of $I_0$ and $S_1$ to obtain the following corollary, from which the proof of Theorem \ref{thm:scc} follows naturally.
	\begin{corollary}\label{cor:no-subsets-without-xapplications}
		Fix $\epsilon > 0$. With high probability, there does not exist a subset of students $I_0\subseteq I$ and a subset of schools $S_1 \subseteq S$ such that (1) both subsets are of size at least $\epsilon n$, and (2) no one in $I_0$ ever applied to any school in $S_1$ in DA.
	\end{corollary}
	\begin{proof}
		For any $I_0\subseteq I$ and $S_1 \subseteq S$ satisfying the cardinality condition, let $\mathcal{E}_{I_0,S_1}$ denote the event that no one in $I_0$ ever applied to any school in $S_1$ in DA. Proving the corollary reduces to bounding the probability of the union event 
		\begin{equation*}
			p := \Pr\bigg(\bigcup_{|I_0|,|S_1|\ge \epsilon n}\mathcal{E}_{I_0,S_1}\bigg).
		\end{equation*}
Let $\mathcal{A}_{I_0}$ denote the event that students in $I_0$ make a combined number of at least $L n$ applications (including redundant ones), for some constant $L$ to be specified later.
		By a union bound, we have
	\begin{equation}\label{eq:union-bound-prob}
		p \le
		\sum_{|I_0|,|S_1|\ge \epsilon n}
		\Pr(\mathcal{E}_{I_0,S_1}\cap \mathcal{A}_{I_0})
		+
		\Pr\left(
		\bigcup_{|I_0|\ge \epsilon n}\mathcal{A}_{I_0}^c
		\right).
	\end{equation}

	Thus, the summation term in the right-hand side of expression \eqref{eq:union-bound-prob} is at most
	\[
	(2^n)^2 \cdot (1-\epsilon)^{Ln} 
	= \exp\!\big((2\log 2 + L \log(1-\epsilon))n\big).
	\]
	This tends to zero as $n \to \infty$, provided we choose $L=L(\epsilon)$ 
	large enough so that $2\log 2 + L \log(1-\epsilon) < 0$.
	
	By Lemma~\ref{lem:many-applications}, applied with $\epsilon/2$ and with $\lceil 2L/\epsilon\rceil$ in place of $L$, with high probability all but at most $\epsilon n/2$ students each make at least $2L/\epsilon$ applications. Hence any subset $I_0$ with $|I_0|\ge \epsilon n$ contains at least $\epsilon n/2$ such students, and therefore students in $I_0$ make at least
	\[
	(2L/\epsilon)\cdot(\epsilon n/2)=Ln
	\]
	applications in total. Thus,
	\[
	\Pr\left(\bigcup_{|I_0|\ge \epsilon n}\mathcal{A}_{I_0}^c\right)\to0.
	\]
		Combining the bounds on the two terms gives an upper bound on $p$ that 	tends to zero as $n\to\infty$, finishing the proof.
	\end{proof}
	
	We now prove Theorem~\ref{thm:scc} using Corollary~\ref{cor:no-subsets-without-xapplications}.
	
\begin{proof}[Proof of Theorem \ref{thm:scc}]
	Note that DA with redundant applications produces the exact same matching and envy graph (modulo duplicate edges due to redundancy).
	Thus, we may use DA with redundant applications as our underlying data generation algorithm.
	Given that $\epsilon$ is arbitrary, it suffices to show that the largest SCC contains at least $(1-2\epsilon)n$ students with high
	probability.
	Consider the condensation DAG of $G^{\da(P_n)}$, whose vertices are the strongly connected components of $G^{\da(P_n)}$. Suppose,
	toward a contradiction, that no SCC contains at least $(1-2\epsilon)n$ students. Take a topological ordering of the condensation
	DAG. Let $B$ be the union of the first SCCs in this ordering until its size first reaches $\epsilon n$. Since each SCC has size
	less than $(1-2\epsilon)n$, both $B$ and $I\setminus B$ have size at least $\epsilon n$. Moreover, by the topological ordering,
	there is no edge from $I\setminus B$ to $B$.
	Set $I_0=I\setminus B$ and $I_1=B$, and let
	\[
	S_1:=\{\da_i(P_n):i\in I_1\}.
	\]
	Then $|I_0|,|I_1|\ge\epsilon n$, $|S_1|\ge\epsilon n$, and no student in $I_0$ has applied to any school in $S_1$. This contradicts
	Corollary~\ref{cor:no-subsets-without-xapplications} with probability tending to one. Hence, with high probability, some SCC
	contains at least $(1-2\epsilon)n$ students. Since $\epsilon$ is arbitrary, the theorem follows.
\end{proof}
	Theorem \ref{thm:scc} yields a surprising Corollary, namely:
	\begin{corollary}
		The fraction of students who are unimprovable converges to zero in large markets, and the DA mechanism produces a Pareto-inefficient matching with high probability. 
	\end{corollary}
	
	Theorem \ref{thm:scc} implies that the set of improvable students forms a substantial majority of the entire student population, but is silent regarding how many of them can be improved by the same mechanism, nor does it specify which mechanism chooses the maximum improvement over DA. Our next Theorem tackles both questions.
	
	\subsection{Equivalence among Mechanisms that Dominate DA}
	
	We define a \textit{cycle packing} $H$ as a union of pairwise-disjoint cycles in $G$, with the property that no additional cycles exist in the subgraph $G\backslash H$. $V(H)$ denotes the vertex set of $H$, or equivalently, the set of nodes that are in some cycle in $H$. Note that, by Proposition \ref{thm:carachterization}, every efficient mechanism corresponds to a cycle packing of DA's envy digraph. Our second Theorem shows that, with high probability, every efficient mechanism improves almost all nodes.
	
	\begin{theorem}\label{thm:cyclepacking}
		With high probability, every cycle packing of the envy digraph $G^{\da(P_n)}$ covers at least $(1-\epsilon)n$ nodes, for any arbitrary constant $\epsilon>0$.
	\end{theorem}
	
	\begin{proof}
		Let $H\subseteq G^{\da(P_n)}$ be a cycle packing. Consider the induced subgraph $G^{\da(P_n)}\backslash H = G^{\da(P_n)}[I\backslash V(H)]$, which by definition, contains no cycles.
		
		Suppose, for contradiction, that $|V(H)| < (1-\epsilon)n$. Then, the directed acyclic subgraph $G^{\da(P_n)}\backslash H$ contains $N = |I\backslash V(H)| > \epsilon n$ vertices. Hence, there exists a topological ordering $i_1, \dots, i_N$ of vertices in $I\backslash V(H)$ such that no directed edge from $i_k$ to $i_\ell$ for $1 \leq k < \ell \leq N$ is present in $G^{\da(P_n)}\backslash H$ (and hence $G^{\da(P_n)}$). 
		
		Partitioning these nodes into $I_0 = \{i_1,\dots,i_{\lfloor N/2 \rfloor}\}$ and $I_1 = \{i_{\lfloor N/2 \rfloor + 1}, \dots, i_N\}$ yields two subsets, each with size at least $\lfloor\epsilon n/2\rfloor$, and crucially, no edges in $G^{\da(P_n)}$ go from $I_0$ into $I_1$. This scenario is identical to the one in Theorem \ref{thm:scc}, which is improbable due to Corollary~\ref{cor:no-subsets-without-xapplications}. Consequently, we must have $|V(H)| \geq (1-\epsilon)n$ with high probability.
	\end{proof}
	
	Theorem \ref{thm:cyclepacking} establishes not merely the prevalence of improvable students but also demonstrates that nearly all such students can concurrently benefit from improvement via non-intersecting trading cycles. 
	Furthermore, Theorem \ref{thm:cyclepacking} establishes the asymptotic equivalence across all such mechanisms regarding the number of students they improve. This result has significant implications for policy implementation. 
	When policymakers deliberate on which efficiency-enhancing mechanism to adopt as an alternative to DA, concerns about differential distributions of who gains by the adoption become largely irrelevant, because all such mechanisms ultimately benefit approximately the same number of students in large markets.
	This equivalence result is particularly striking given that in carefully constructed small instances, different efficient mechanisms that dominate DA may significantly differ in their corresponding number of improved students \citep{knipe2025improvable}. However, we want to stress that our equivalence results do not imply that the size of the improvement is the same across efficient mechanisms that weakly Pareto-dominate DA, nor their equivalence with regard to stability criteria such as number of blocking pairs.\footnote{However, all efficient mechanisms (whether they dominate DA or not) have an asymptotically equivalent \emph{normalized} average rank \citep{che2018payoff}.}

	\subsection{Simulations and Extensions}
	
	The simulations depicted in Figure \ref{fig:scc1} show the emergence of a unique giant component in $G^{\da(P_n)}$ for even relatively small values of $n$.
	\begin{figure}[h!]
		\centering
		\includegraphics[width=\textwidth]{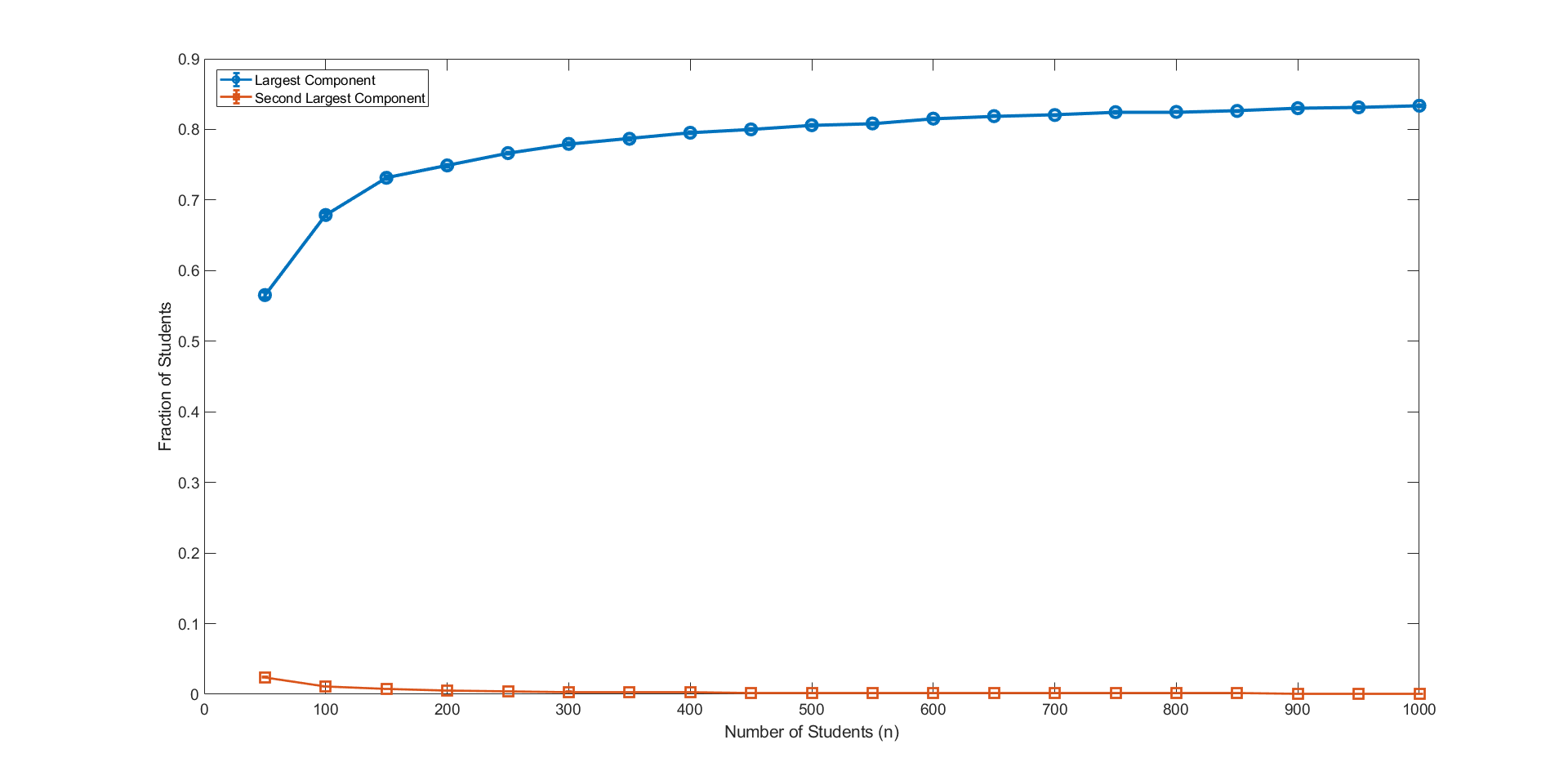} 
		\caption{Average fraction of nodes in the largest (blue) and second-largest (red) SCCs (average over 2,000 random envy digraphs for each $n$).}
		\label{fig:scc1}
	\end{figure}
	
	The simulations in Figure \ref{fig:unimp2} illustrate the convergence of unimprovable students to zero as market size increases. The convergence rate is approximately $O(1/\log n)$, directly corresponding to the rate at which students receive their top choice under DA. This theoretical relationship explains the gradual nature of the decline.
	\begin{figure}[h!]
		\centering
		\includegraphics[width=\textwidth]{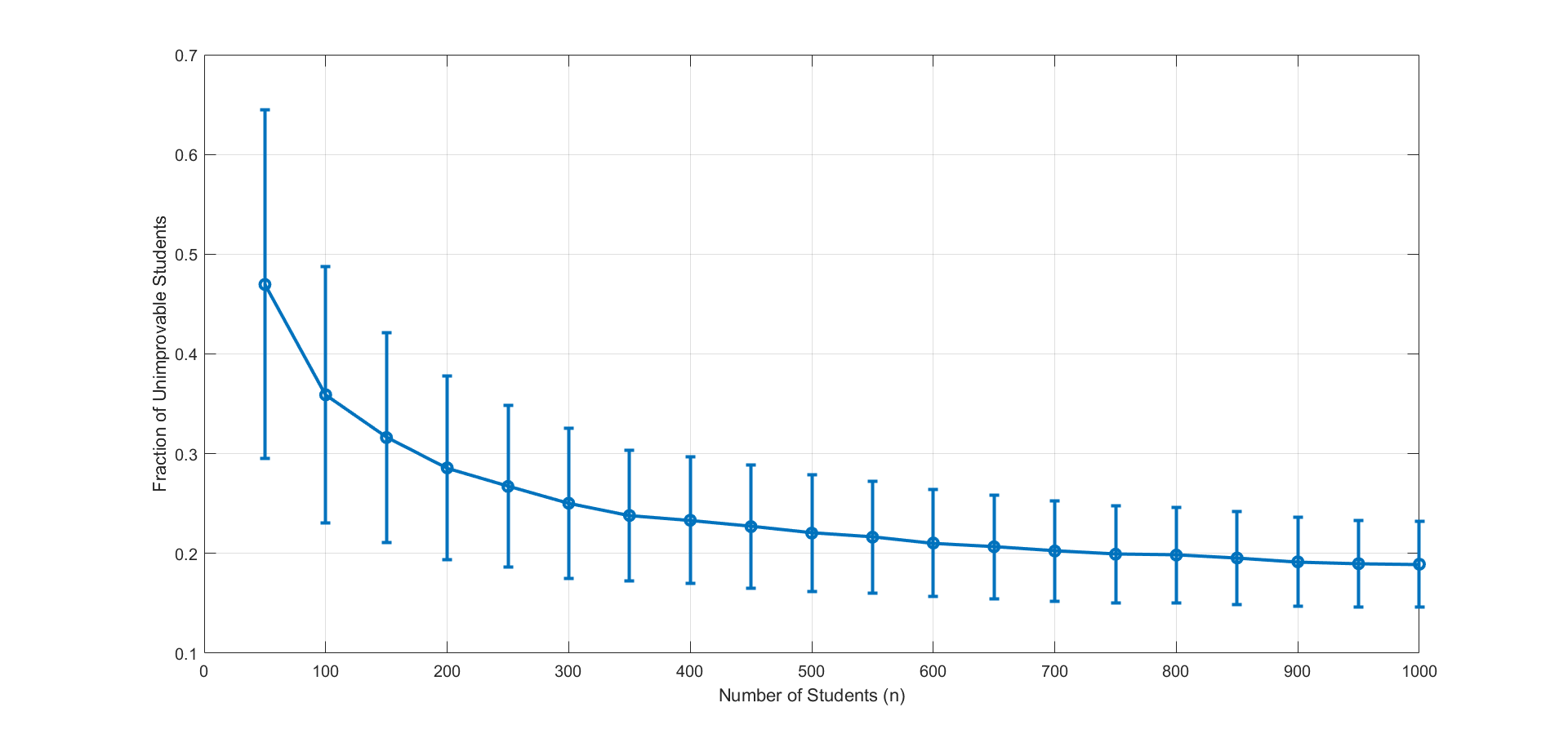} 
		\caption{Average fraction of unimprovable students (average over 2,000 random problems for each $n$).}
		\label{fig:unimp2}
	\end{figure}
	
	In the Supplementary Appendix, we discuss \new{three} extensions of our work.
	\begin{enumerate}
		\item Many-to-one Matching: Our probabilistic analysis has assumed that the number of schools and students is the same, yet our main results (Theorems \ref{thm:scc} and \ref{thm:cyclepacking}) can be extended to a model with $n$ schools, each with a quota $q$, and $qn$ students. In this model, each student makes fewer applications, but each rejection makes him envy $q$ instead of only 1 student. This trade-off between application frequency and envy multiplicity results in an envy digraph that becomes even more densely connected. Consequently, the emergence of a unique giant strongly connected component remains a robust phenomenon, as well as the asymptotic size equivalence among efficient mechanisms that dominate DA. However, convergence occurs more slowly as $q$ increases. We provide proofs and simulations in the Supplementary Appendix.
		
		\item Correlated Preferences:
			In Appendix \ref{app:appendixtwo}, we show that our results are robust to correlated preferences well beyond the i.i.d.\ uniform benchmark. We first consider tiered preferences in the spirit of \citet{ashlagi2020tiered}: schools are partitioned into a finite number of quality tiers, and preferences are independent within each tier. We prove that the envy digraph then contains one giant strongly connected component within each tier (hence as many giant SCCs as tiers), and simulations indicate that the union of these giant components covers an asymptotic fraction of students close to one.
			
			We then state sufficient conditions under which the giant-SCC proof extends to general correlated-preference models. These conditions capture the two key ingredients needed: large sets of students continue to make enough proposals, and their top choices expand uniformly across large sets of schools. We show that bounded rank-ordered logit preferences satisfy both conditions, and hence the giant-SCC result follows for bounded rank-ordered logit markets. Finally, simulations under additive common-and-idiosyncratic utilities show that moderate correlation can increase the size of the giant SCC relative to independence, reinforcing that the pervasiveness of improvable students is not an artifact of iid preferences.
		
\item \new{Correlated Priorities: In Appendix \ref{app:appendixthree}, we use simulations to examine the role of correlation in schools' priorities. Unlike correlation in preferences, which tends to strengthen our results, correlation in priorities can reduce the size of the (still unique) giant SCC. The intuition is that correlated priorities create a de facto hierarchy among students: high-priority students are accepted quickly with few applications, while rejections concentrate on low-priority students, reducing the potential for cycles. Our simulations show that the giant SCC remains substantial even under moderate priority correlation, with over half of students remaining improvable. When both preferences and priorities are moderately correlated, the preference effect dominates and the giant SCC remains larger than under the iid baseline.\footnote{\new{It is well-known that if all schools share the same priorities, or more generally if priorities are acyclic, DA is Pareto efficient \citep{ergin2002efficient,kesten2006two}. In Appendix \ref{app:role-priorities}, we show that the iid-priority benchmark is far from this acyclic case: the probability of Ergin acyclicity vanishes super-exponentially fast.}}}
	\end{enumerate}

\subsection{Scope and Interpretation}
Our results show that stable mechanisms fail even the weakest efficiency requirement: in large markets, almost all students can be made strictly better off without harming anyone else. This highlights a fundamental inefficiency of stability. At the same time, our framework is intentionally silent about the magnitude of the gains. Because we work in the standard Gale--Shapley model with purely ordinal preferences, we do not evaluate how large the improvements are in utilitarian terms. Addressing that question requires additional structure on cardinal utilities and aggregation, as in \citet{lee2014efficiency} and \citet{che2019efficiency}. Our findings should therefore be interpreted as complementary to that literature. We establish that Pareto improvements are pervasive, while the size of their welfare impact depends on the underlying cardinal environment.

Finally, allowing indifferences in preferences mainly affects whether improvements are weak or strict. In such environments, some students in a trading cycle may move within the same indifference class, while others strictly gain. Our results continue to imply Pareto inefficiency with high probability, but the exact decomposition between weak and strict improvements depends on the structure generating ties, and lies beyond the scope of our ordinal model. A similar remark applies to priority ties. When priorities must be broken by lottery, different tie-breaking rules generate different stable matchings, and some students may end up strictly worse than under alternative tie breaks \citep{erdil2008s}. In this sense, our baseline model with strict priorities provides a conservative lower bound on the inefficiencies that arise in practice.
	
\subsection{Empirical Perspective}
Since any student receiving her top choice is necessarily unimprovable, empirical top-choice shares provide a natural lower bound on the fraction of unimprovable students in real markets. These vary widely across markets. 
Examples include
Barbados (6–13\%; \citealp{pariguana}), New York City (41\%; \citealp{abdulkadirouglu2009strategy}), Budapest (47\%; \citealp{ortega2023cost}), Chile (52\%;  \citealp{correa2022school}), and Boston (73\%;  primarily for primary schools; \citealp{abdulkadiroglu2006changing}).

However, these statistics are strongly affected by self-selection in submitted rank lists \citep{chen2019self,pariguana}. When true preferences are recovered, top-choice shares are substantially lower---for example, 11\% in Budapest \citep{ortega2023cost}---which is consistent with our asymptotic predictions.
Direct estimates of the fraction of students lying on cycles in the DA envy digraph are largely absent from existing data, highlighting a promising direction for future empirical work.\footnote{\cite{abdulkadirouglu2009strategy} find that roughly 2\% of students in NYC could improve their DA placement via stable improvement cycles. However, this figure reflects only students who obtain a better placement under one specific mechanism that dominates DA but which does not guarantee full Pareto efficiency.}

\section{Conclusion}
\label{sec:conclusion}

We show that Deferred Acceptance (DA) is highly Pareto inefficient in large markets: as the market grows, trading cycles in the envy digraph cover almost all students, so that nearly every student becomes improvable relative to the DA outcome.
Because every Pareto improvement upon DA is also a Pareto improvement upon any other stable matching (since DA selects the student-optimal stable outcome) this conclusion immediately extends beyond DA. Hence, our results imply that in large markets all stable mechanisms admit Pareto improvements for almost all students.

Finally, while these findings motivate considering efficiency adjustments that relax stability, our outcome-equivalence theorem indicates that such debates should be guided by properties other than the number of beneficiaries: leading stable-dominating approaches generate essentially the same set of improved students in large markets, so the relevant distinctions lie in other institutional and normative criteria.

\section*{Acknowledgements}
We are grateful to three anonymous reviewers and the Editor for their detailed comments which have improved our paper. We also acknowledge helpful comments from Mariagiovanna Baccara, Li Chen, David Delacrétaz, Battal Doğan, Matt Elliott, Arda Gitmez, Thilo Klein, Jörgen Kratz, SangMok Lee, Vikram Manjunath, Martin Meier, Alexander Nesterov, Marco Pariguana, Szilvia Pápai, Fedor Sandomirskiy, Qianfeng Tang, Olivier Tercieux, Bertan Turhan, M. Utku Ünver, Alexander Westkamp, and M. Bumin Yenmez.

	\singlespacing      
	\setlength{\parskip}{-0.2em} 
	\bibliographystyle{te} 
	\bibliography{bibliogr}
	\newpage
	\appendix
	\setcounter{secnumdepth}{3}
	\section*{}\label{appendix}
	\subsection{Extension to Many-to-One Matching}
\label{app:appendixone}	
	In the main text, we have shown that a giant unique SCC appears in $G^{\da(P_n)}$ if there are as many schools as students. We relax such assumption here, allowing for a school to be able to admit many students, as follows.
	
	\begin{theorem}
		\label{thm:many-to-one-scc}
		Consider a random school choice problem $P_n$ with $n$ schools, where each school has the same fixed quota $q$, and $nq$ students. Fix an arbitrary $\epsilon > 0$. With high probability, the envy digraph $G^{\da(P_n)}$ has a unique giant strongly connected component containing at least $(1 - \epsilon)nq$ students.
	\end{theorem}
	
	\begin{proof}
		We follow a strategy similar to Theorem \ref{thm:scc}, adapting it to the many-to-one setting. First, we establish bounds on the number of applications in the DA algorithm, then analyze partitions of the resulting envy digraph.
		
		\paragraph{Step 1: Bounds on applications.} 
		
		We first establish the following bound on the number of applications a typical student makes, similar to Lemma~\ref{lem:many-applications}:
		
		\begin{lemma}\label{lem:many-applications-many-to-one}
			Fix $\epsilon>0$.
			With high probability, all but at most $\epsilon nq$ students each make at least $\sqrt{\log n}$ applications.
		\end{lemma}
		
		\begin{proof}
			Again, it suffices to show that for a fixed student $i$, the number of applications she makes is at least $\sqrt{\log n}$ with high probability; the lemma follows as a consequence of Markov's inequality applied to the number of students making fewer applications than this, same as we saw in the proof of Lemma~\ref{lem:many-applications}.
			
			The standard approach for one-to-one matchings using the coupon collector problem as an approximation can be adapted to this case with slight modifications.
			In particular, one can establish that the total number of applications throughout DA is of order $\Theta(n \log n)$. 
			Since the order of applications in DA is irrelevant, one can hold off $i$ and run DA on the rest of the students $I\backslash\{i\}$, and when this first stage terminates (i.e., when all the students except for $i$ are stably matched) there should be already $\Theta(n \log n)$ applications. At this point, all but one school has their capacity filled; Without loss of generality, let it be school $s_n$. Using the terminology from \cite{ashlagi2020tiered}, we say that this intermediate state of DA is \emph{smooth} if (1) at most $O(n\log n)$ applications are made until this point and (2) all but at most $n^a$ schools have received applications from $\Omega(\log n)$ distinct students with constant $a<1$. With nearly identical analysis as in \citet{ashlagi2020tiered}, we can show that a smooth state is achieved with high probability; the details are repetitive and hence omitted.\footnote{The first bound on the number of applications can be derived from classic results on the generalized coupon collector problem; see, e.g., \cite{newman1960double,erdHos1961classical}. The second bound on the number of applications received by the schools is slightly more involved. To derive it, we run the amnesiac DA \citep{knuth1976}, and using concentration arguments on the number of draws in (generalized) coupon collector problem, one can show that the number of applications, counting duplicates, is of order $\Omega(n\log n)$ up to this stage; further, duplicates are rare among them, so the total number of applications in the classic DA is again $\Omega(n\log n)$ with high probability prior to the second stage. The number of applications each school receives can then be bounded using a Poissonization argument or with direct Chernoff-type bounds. See, e.g., \citet[Proposition~4.3]{ashlagi2020tiered} for a formal statement and analysis in the one-to-one case.}
			We now
			let $i$ start making applications until the full DA terminates (i.e., someone applies to this last unfilled school $s_n$). Note that, by symmetry, an application to a school that has received $m\ge q$ applications will be accepted with probability $q/(m+1)$. Denote the number of applications to school $s_k$  at the beginning of this second stage by $M_k$ for $k=1,\ldots,n-1$. Under a smooth state, we may assume without loss of generality that $M_k \ge \Theta(\log n)$ for $k=1,\ldots,\lfloor n-n^a\rfloor$; as the second stage of DA unfolds, the number of applications to each school can only increase.
		The probability that an application from $i$ is accepted is at most 
		\begin{equation*}
			\frac{1}{n}\Big(1 + \sum_{k=1}^{n-1} \frac{q}{M_k + 1}\Big)
			\le 
			\frac{1+n^a}{n}
			+
			\frac{n-n^a}{n}\cdot\frac{q}{\Theta(\log n)}
			=
			O(q/\log n).
		\end{equation*}
		Hence, for all sufficiently large $n$, each application by $i$ is accepted with probability at most $Cq/\log n$ for some constant $C>0$. Conditional on the likely smooth state before $i$ starts to apply, the number $T_i$ of applications that $i$ makes until receiving a tentative acceptance is therefore stochastically lower bounded by a geometric random variable with success probability $Cq/\log n$. Thus,
		\[
		\Pr(T_i<\sqrt{\log n}\mid \text{smooth})\le 1-(1-Cq/\log n)^{\sqrt{\log n}}\to0.
		\]
		Since a smooth state is achieved with high probability, the marginal probability of $i$ making fewer than $\sqrt{\log n}$ applications vanishes as $n\to\infty$, finishing our proof.
		\end{proof}

		\paragraph{Step 2: Partition Argument.} 
		Suppose $G^{\da(P_n)}$ has two disjoint subsets $I_0, I_1 \subseteq I$, each of size $\geq \epsilon nq$, with no edges from $I_0$ to $I_1$. Let $S_1 = \{s \in S: \exists j \in I_1 \text{ with } \da_j(P_n)=s\}$, so $|S_1| \geq \epsilon n$. We now show that for any fixed subsets $I_0$ and $S_1$, it is extremely unlikely that no one in $I_0$ ever applies to schools in $S_1$.
		
		\begin{lemma}\label{lem:no-crossapplication-many-to-one}
			For any fixed $I_0 \subseteq I$ and $S_1 \subseteq S$ with $|I_0| \geq \epsilon nq$ and $|S_1| \geq \epsilon n$, and for any fixed constant $c>0$, the probability that students in $I_0$ make at least $c n\sqrt{\log n}$ applications yet none apply to $S_1$ is at most $(1-\epsilon)^{c n\sqrt{\log n}}$.
		\end{lemma}
		
		\begin{proof}
			We use the same idea as in the proof of Lemma~\ref{lem:no-crossapplication}. In the amnesiac DA model, each application independently targets a school in $S_1$ with probability at least $\epsilon$. Thus, the probability that $c n\sqrt{\log n}$ applications all avoid $S_1$ is at most $(1-\epsilon)^{c n\sqrt{\log n}}$, as claimed.
		\end{proof}
		
		\paragraph{Step 3: Union Bound over Partitions.} 
		Fix $\eta>0$. The event that the largest SCC contains strictly less than $(1-2\eta)nq$ students implies the existence of linear-sized $I_0,I_1\subseteq I$ with $|I_0|,|I_1|\ge \eta nq$ and with no edges going from $I_0$ to $I_1$. Thus, no student in $I_0$ ever applies to the set of schools
		\[
		S_1:=\{s\in S:\exists j\in I_1,\ \da_j(P_n)=s\}.
		\]
		Since each school has capacity $q$, we have $|S_1|\ge \eta n$. Denote the no-application event by $\mathcal{E}_{I_0,S_1}$.
		
		Set $c_{\eta,q}:=\eta q/2$. Let $\mathcal{A}_{I_0}$ denote the event that the subset of students $I_0$ makes at least $c_{\eta,q}n\sqrt{\log n}$ applications in total. By a union bound,
		\begin{equation}
			p \leq 
			\sum_{|I_0|\geq \eta nq,\ |S_1|\ge \eta n}
			\Pr(\mathcal{E}_{I_0,S_1} \cap \mathcal{A}_{I_0})
			+
			\Pr\left(\bigcup_{|I_0|\ge \eta nq} \mathcal{A}_{I_0}^c\right).
		\end{equation}
		
		By Lemma~\ref{lem:no-crossapplication-many-to-one}, the first summation term is at most
		\begin{equation}
			2^{nq+n} \cdot (1-\eta)^{c_{\eta,q}n\sqrt{\log n}}
			=
			\exp\left((q+1)n\ln 2+c_{\eta,q}n\sqrt{\log n}\log(1-\eta)\right)\to0.
		\end{equation}
		
		Applying Lemma~\ref{lem:many-applications-many-to-one} with $\eta/2$, with high probability all but at most $\eta nq/2$ students make at least $\sqrt{\log n}$ applications. Hence every subset $I_0$ with $|I_0|\ge \eta nq$ contains at least $\eta nq/2$ students who each make at least $\sqrt{\log n}$ applications, so $\mathcal{A}_{I_0}$ holds for every such $I_0$. Therefore,
		\[
		\Pr\left(\bigcup_{|I_0|\ge \eta nq} \mathcal{A}_{I_0}^c\right)\to0.
		\]
		
		Thus, with high probability, the largest SCC contains at least $(1-2\eta)nq$ students. Since $\eta>0$ is arbitrary, the theorem follows.
	\end{proof}
	
	Despite the different structure of the many-to-one market with uniform quotas, the fundamental conclusion remains: in large random markets, a unique giant SCC emerges, and thus most students are improvable.
	However, we emphasize that the rate of convergence in the many-to-one setting differs from the one-to-one case, and depends on the size of the constant $q$: As $q$ increases, each school forms larger ``envy-free clusters'' (students assigned to the same school), potentially slowing convergence.
	
	In Figures \ref{fig:q5} and \ref{fig:q10}, we observe exactly that: the unique giant SCC emerges but the fraction of nodes in it for fixed $n$ becomes smaller for larger constant $q$.
	
	Note that although the convergence appears slow, scaling like $1/\log n$ and consistent with our analysis, the main driver behind this slow rate is the highly dispersed distribution of applications submitted by students. In particular, an $O(1/\log n)$ fraction of students submit only one application and thus envy no one. Empirical simulations suggest that, aside from these trivially unimprovable students, the giant strongly connected component (SCC) contains nearly all of the remaining students.
	
	\begin{figure}[htbp]
		\centering
		\begin{subfigure}{0.48\textwidth}
			\centering
			\includegraphics[width=\textwidth]{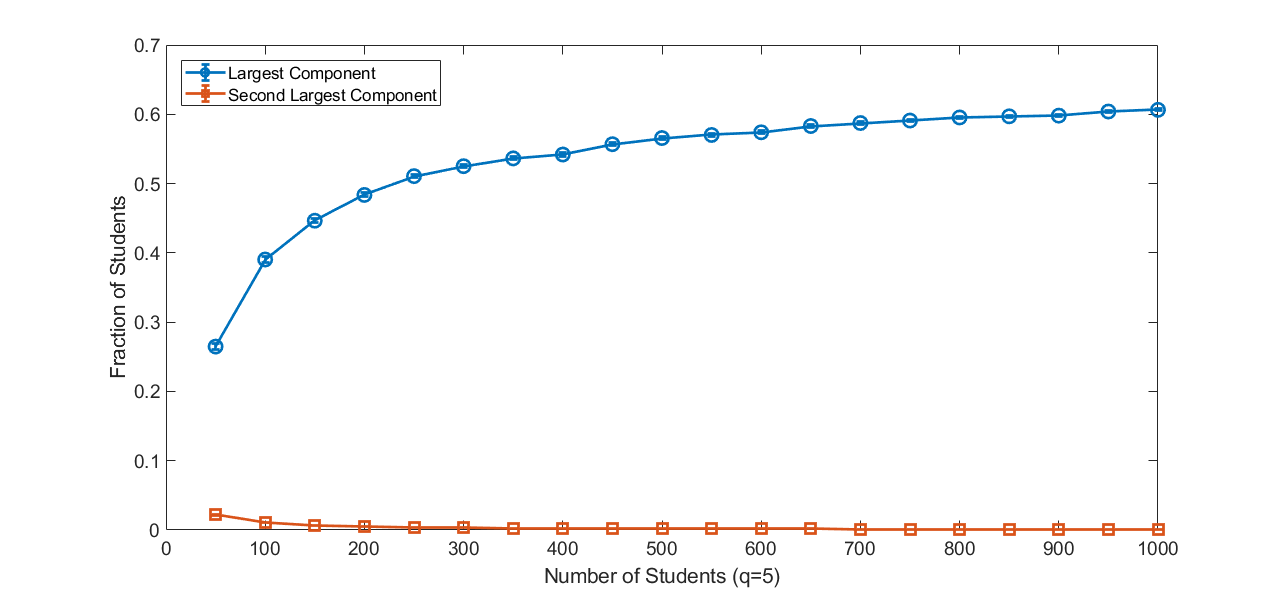}
			\caption{$n=[10:10:200], q=5$.}
			\label{fig:q5}
		\end{subfigure}
		\hfill
		\begin{subfigure}{0.48\textwidth}
			\centering
			\includegraphics[width=\textwidth]{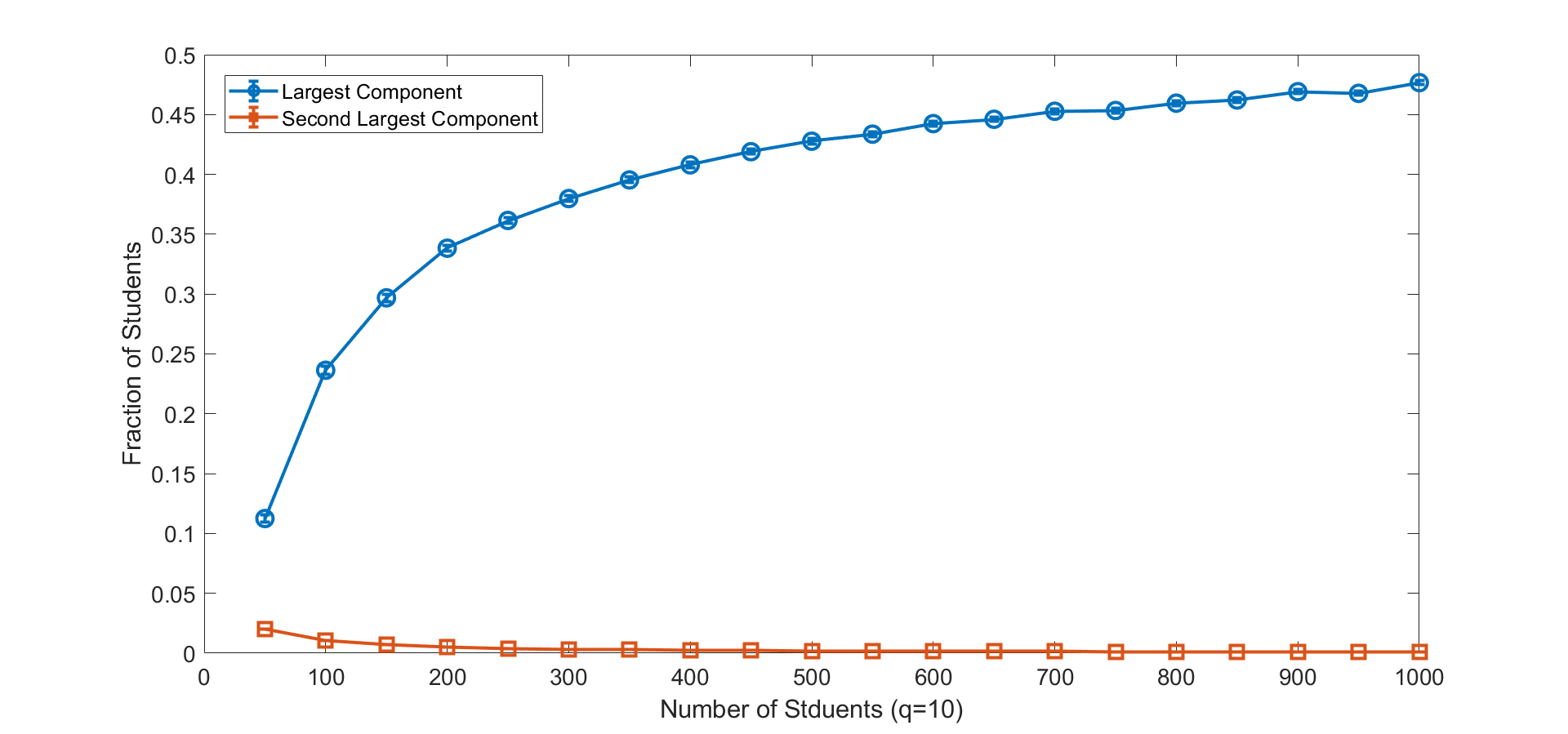}
			\caption{$n=[5:5:100], q=10$.}
			\label{fig:q10}
		\end{subfigure}
		\caption{Average fraction of nodes in the largest (blue) and second-largest (red) SCCs (average over 1,000 random problems for each $n$), preferences and priorities uniform i.i.d.}
		\label{fig:comparison}
	\end{figure}

\paragraph{Equivalence in Many-to-One Matching.}

Theorem \ref{thm:cyclepacking} also extends naturally to the many-to-one setting with fixed quota $q$ for each school. We present this generalization as follows:

\begin{theorem}
	\label{thm:cyclepacking-many-to-one}
	Consider a random school choice problem $P_n$ with $n$ schools, where each school has the same fixed quota $q$, and $nq$ students. Fix an arbitrary $\epsilon > 0$. With high probability, every cycle packing of the envy digraph $G^{\da(P_n)}$ covers at least $(1-\epsilon)nq$ students.
\end{theorem}

\begin{proof}
	The proof follows the approach of Theorem \ref{thm:cyclepacking}, adapted to the many-to-one setting. Let $H \subseteq G^{\da(P_n)}$ be a cycle packing, and consider the induced subgraph $G^{\da(P_n)} \backslash H = G^{\da(P_n)}[I \backslash V(H)]$, which contains no cycles by definition.
	
	Suppose, for contradiction, that $|V(H)| < (1-\epsilon)nq$. Then, the directed acyclic subgraph $G^{\da(P_n)} \backslash H$ contains $N = |I \backslash V(H)| > \epsilon nq$ vertices. We can find a topological ordering $i_1, \dots, i_N$ of vertices in $I \backslash V(H)$ such that no directed edge from $i_k$ to $i_{\ell}$ for $1 \leq k < \ell \leq N$ is present in $G^{\da(P_n)} \backslash H$.
	
	Partitioning these nodes into $I_0 = \{i_1,\dots,i_{\lfloor N/2 \rfloor}\}$ and $I_1 = \{i_{\lfloor N/2 \rfloor + 1}, \dots, i_N\}$ yields two subsets, each with size at least $\lfloor\epsilon nq/2\rfloor$, with no edges going from $I_0$ to $I_1$ in $G^{\da(P_n)}$.
	
	This scenario implies that no student in $I_0$ has applied to any school matched to students in $I_1$ during the execution of DA. Let
	\[
	S_1:=\{s\in S:\exists j\in I_1,\ \da_j(P_n)=s\}.
	\]
	Since each school has capacity $q$, for all sufficiently large $n$ we have $|S_1|\ge \epsilon n/3$. By the union-bound argument in the proof of Theorem~\ref{thm:many-to-one-scc}, applied with $\eta=\epsilon/3$, the probability that such subsets exist tends to zero.
		
	Thus, we must have $|V(H)| \geq (1-\epsilon)nq$ with high probability.
\end{proof}
These findings collectively demonstrate that our main conclusions---both about the pervasiveness of improvable students and the asymptotic equivalence of mechanisms that Pareto dominate DA---are robust to varying market structures, including the practically important case of schools with multiple seats.


\subsection{Correlation in Preferences}
\label{app:appendixtwo}
In the main text, we assume that students' preferences are independently drawn. This assumption may not hold in some scenarios where students share common views about school quality while retaining individual tastes. Here we relax the independence assumption and establish that our main results are robust to realistic correlation structures. We first provide two theoretical extensions under two correlation models. Then, we provide extensive simulation analysis to show that our results extend even beyond these correlation structures.

\subsubsection{Model 1: Tiered Preferences}

We first consider a tiered preference model where schools are divided into tiers, and every school in tier $k$ is preferred to any school in tier $k+1$. This captures the empirical observation that some schools are widely recognized as higher quality. Within tiers, preferences are drawn uniformly at random.\footnote{Tiered preference and priority structures in matching markets have been discussed in \cite{allman2025signaling,ashlagi2017,ashlagi2020tiered,ccelebi2022priority,calsamiglia2023catchment, cai2025manipulability, dur2016explicit}.}

\begin{theorem}
	\label{thm:tiered}
	Consider a market with tiers of schools of equal size, where all schools in tier $k$ are preferred to all schools in tier $k+1$ by all students. Within each tier, preferences are drawn uniformly at random, and priorities are drawn uniformly at random across all schools. Then as $n \to \infty$:
	\begin{enumerate}
	\item There exist exactly $c$ giant strongly connected components in the DA envy digraph, one per tier, and no strongly connected component contains students assigned to schools in different tiers.
		\item Within each tier $k$, the fraction of tier-$k$ students in that tier's giant SCC converges to 1 in probability.
	\end{enumerate}
\end{theorem}

\begin{proof}
	The tiered structure induces complete segregation: a student matched to tier $k$ can only envy students matched to tiers $1, \ldots, k$, and can only be envied by students matched to tiers $k, \ldots, c$. Therefore, cycles cannot span across tiers, implying that strongly connected components must lie entirely within single tiers (see Theorem 2 in \citealt{ortega2025pareto}).
	
Formally, let $I_k$ denote the set of students matched to tier-$k$ schools under DA, and let $\mathcal{G}_k$ denote the induced subgraph of the envy digraph restricted to students in $I_k$.
Within each tier $k$, conditional on the set $I_k$, the random market structure is identical to the i.i.d. model analyzed in Theorem~\ref{thm:scc}. Each student in $I_k$ has preferences drawn uniformly over tier-$k$ schools, independently of her preferences over other tiers, and priorities are uniformly random. Therefore, the analysis from Theorem~\ref{thm:scc} applies directly to $\mathcal{G}_k$, establishing that a giant SCC emerges within tier $k$ containing a fraction converging to 1 of students in $I_k$.
\end{proof}

For the two-tier case, simulations confirm that two giant SCCs emerge, together containing essentially all students as the market grows (Figure \ref{fig:scc2}).

\begin{figure}[h!]
	\centering
	\includegraphics[width=\textwidth]{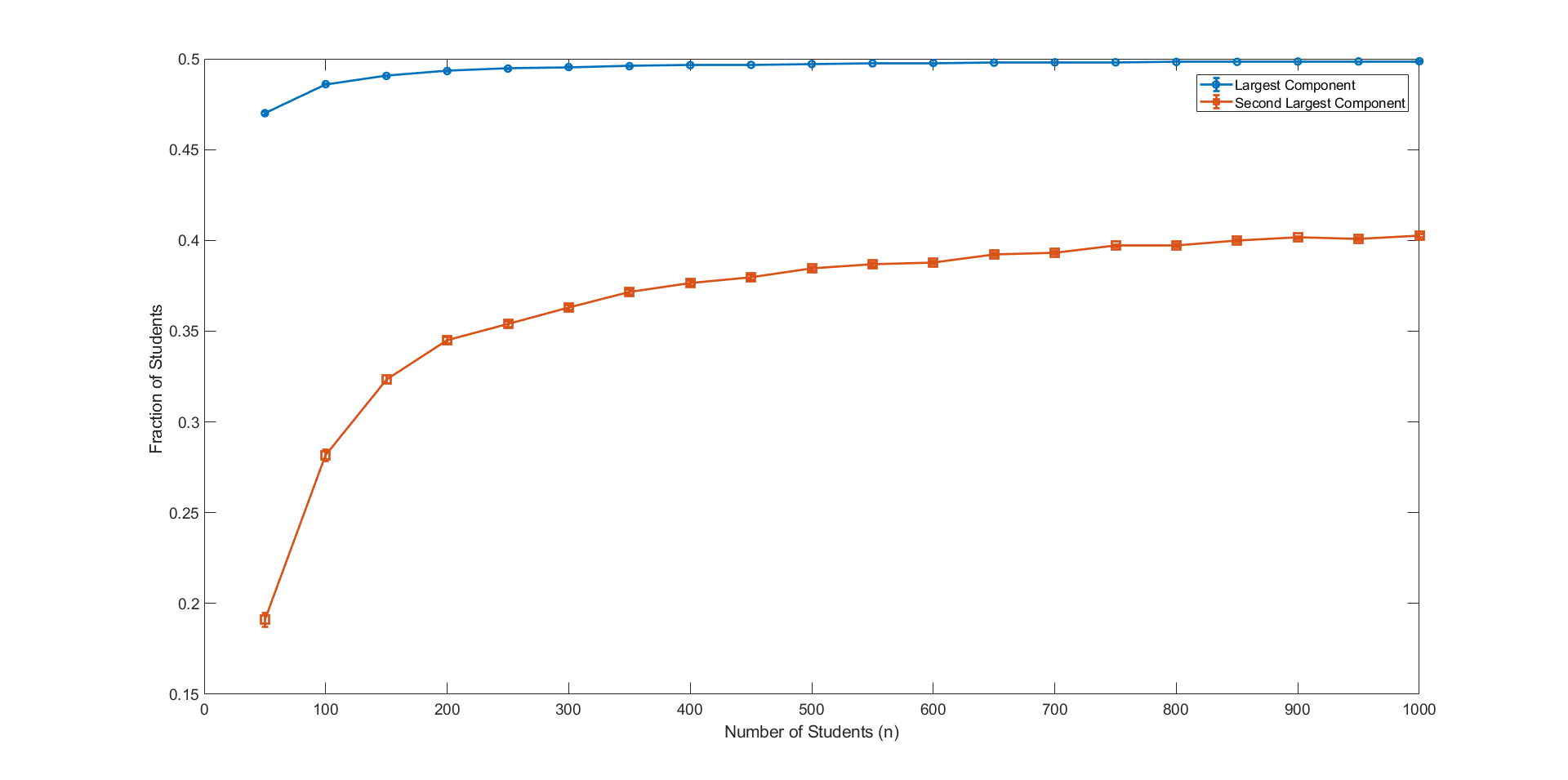} 
	\caption{Average fraction of nodes in the largest (blue) and second-largest (red) SCCs in a two-tiered model (average over 1,000 random envy digraphs for each $n$).}
	\label{fig:scc2}
\end{figure}

\subsubsection{Model 2: General Correlated Preferences}

We now state sufficient conditions under which the giant-SCC result extends beyond the i.i.d.\ uniform benchmark. Throughout this subsection, priorities remain independently and uniformly drawn; only student preferences are allowed to be correlated. The proof of Theorem~\ref{thm:scc} relies on two ingredients. First, large sets of students must not collectively rank large sets of schools too low. Second, enough students must make enough proposals for these preference comparisons to become realized applications. We state these two requirements as uniform high-probability conditions over all linear-sized sets.

\begin{assumption}[Uniform Top-List Expansion]
	\label{ass:entropy}
	For every constant $\delta\in(0,1)$, there exist an integer $L_\delta\ge 1$ and a constant $\rho_\delta>0$ such that, with high probability, the following holds simultaneously for every pair of sets $I_0\subseteq I$ and $S'\subseteq S$ with $|I_0|\ge \delta n$ and $|S'|\ge \delta n$:
	\[
	\left|
	\left\{
	i\in I_0:
	\text{some school in }S'\text{ is among student }i\text{'s top }L_\delta\text{ schools}
	\right\}
	\right|
	\ge 2\rho_\delta |I_0|.
	\]
\end{assumption}

\begin{assumption}[Sufficient Competition]
	\label{ass:proposals}
	For every constant $\eta>0$ and every fixed integer $L\ge 1$, with high probability all but at most $\eta n$ students make at least $L$ proposals during the execution of DA.
\end{assumption}

Assumption~\ref{ass:entropy} is the correlated-preference analogue of the expansion estimate used in the i.i.d.\ proof. It is deliberately stated uniformly over all large sets of students and schools, because the proof must rule out exponentially many possible cuts in the envy graph. Assumption~\ref{ass:proposals} is the analogue of Lemma~\ref{lem:many-applications}: it ensures that top-list preference expansion becomes realized as actual applications during DA. Together, the two assumptions exclude the degenerate case of nearly perfect correlation, in which a large set of schools may be ranked very low by almost all students.

The next lemma combines the two assumptions into the exact application-expansion property needed in the SCC proof.

\begin{lemma}
	\label{lem:uniform-application-expansion}
	Under Assumptions~\ref{ass:entropy} and~\ref{ass:proposals}, for every constant $\delta\in(0,1)$ there exists a constant $\rho_\delta>0$ such that, with high probability, the following holds simultaneously for every pair of sets $I_0\subseteq I$ and $S'\subseteq S$ with $|I_0|\ge \delta n$ and $|S'|\ge \delta n$:
	\[
	\left|
	\left\{
	i\in I_0:
	i\text{ applies to some school in }S'\text{ during DA}
	\right\}
	\right|
	\ge \rho_\delta |I_0|.
	\]
\end{lemma}

\begin{proof}
	Fix $\delta\in(0,1)$. Let $L_\delta$ and $\rho_\delta$ be the constants from Assumption~\ref{ass:entropy}. Apply Assumption~\ref{ass:proposals} with $L=L_\delta$ and $\eta=\rho_\delta\delta$. With high probability, at most $\rho_\delta\delta n$ students make fewer than $L_\delta$ proposals.
	
	On the event in Assumption~\ref{ass:entropy}, for every $I_0,S'$ with $|I_0|,|S'|\ge\delta n$, at least $2\rho_\delta |I_0|$ students in $I_0$ rank some school in $S'$ among their top $L_\delta$ schools. Among these students, at most $\rho_\delta\delta n\le \rho_\delta |I_0|$ make fewer than $L_\delta$ proposals. Hence at least $\rho_\delta |I_0|$ students in $I_0$ both rank some school in $S'$ among their top $L_\delta$ schools and make at least $L_\delta$ proposals. Each such student must apply to some school in $S'$ during DA. This proves the claim.
\end{proof}

\begin{theorem}
	\label{thm:correlated-giant}
	Fix an arbitrary $\epsilon>0$. Under any preference model satisfying Assumptions~\ref{ass:entropy} and~\ref{ass:proposals}, with high probability the DA envy digraph $G^{\da(P_n)}$ has a giant SCC containing at least $(1-\epsilon)n$ students.
\end{theorem}

\begin{proof}
	Fix $\epsilon>0$ and set $\gamma=\epsilon/2$. By Lemma~\ref{lem:uniform-application-expansion}, with high probability, for every pair of sets $I_0\subseteq I$ and $S'\subseteq S$ with $|I_0|,|S'|\ge \gamma n$, some student in $I_0$ applies to a school in $S'$ during DA.
	
	Suppose, toward a contradiction, that no SCC of $G^{\da(P_n)}$ contains at least $(1-\epsilon)n=(1-2\gamma)n$ students. Consider the condensation DAG of $G^{\da(P_n)}$, whose vertices are the strongly connected components of $G^{\da(P_n)}$. As in the proof of Theorem~\ref{thm:scc}, there exist disjoint subsets $I_0,I_1\subseteq I$ with $|I_0|,|I_1|\ge\gamma n$ and no directed edge from $I_0$ to $I_1$.
	
	Let
	\[
	S_1:=\{\da_i(P_n):i\in I_1\}.
	\]
	Since we are in the one-to-one case, $|S_1|=|I_1|\ge\gamma n$. By Lemma~\ref{lem:uniform-application-expansion}, some student $i\in I_0$ applies to some school $s\in S_1$ during DA. Let $j\in I_1$ be the student assigned to $s$ under DA, so that $s=\da_j(P_n)$. Since $i\in I_0$ and $j\in I_1$, student $i$ is not assigned to $s$. Having applied to $s$ before receiving her final DA assignment, student $i$ strictly prefers $s$ to $\da_i(P_n)$. Therefore there is an envy edge from $i$ to $j$, contradicting the construction of $I_0$ and $I_1$.
	
	Hence, with high probability, some SCC contains at least $(1-\epsilon)n$ students. Since $\epsilon>0$ was arbitrary, the largest SCC contains a fraction converging to one in probability.
\end{proof}

We note that the uniformity in Assumption~\ref{ass:entropy} is essential. A condition saying that, for one fixed student and one fixed school set, the student reaches that set with probability $\omega(1/n)$ is not enough: the proof must rule out $\binom{n}{\epsilon n}\binom{n}{\epsilon n}=e^{\Theta(n)}$ possible cuts. This is why the assumption is stated directly as a high-probability expansion condition over all linear-sized sets.

\paragraph{Application to Bounded Rank-Ordered Logit Preferences.}
A natural correlated-preference model satisfying our sufficient conditions is the bounded rank-ordered logit model. Each school $j$ has a public attractiveness weight $\theta_j>0$. Conditional on the vector $(\theta_j)_{j\in S}$, each student's ranking is generated as a sequence of multinomial-logit choices without replacement: at each step, the next-ranked school is selected from the remaining schools with probability proportional to its attractiveness weight.

We assume that the public attractiveness ratio is bounded: after normalization,
\[
1\le \theta_j\le C
\qquad\text{for every }j\in S,
\]
for some constant $C<\infty$ independent of $n$. Priorities remain independently and uniformly drawn.

\begin{lemma}[Top-list expansion]
	\label{lem:logit-toplist}
	For every fixed $C<\infty$, the bounded rank-ordered logit preference model satisfies Assumption~\ref{ass:entropy}.
\end{lemma}

\begin{proof}
	Fix $\delta\in(0,1)$ and condition on a score vector $(\theta_j)_{j\in S}$ satisfying $1\le\theta_j\le C$ for all $j$. Preferences are independent across students conditional on this vector.
	
	Fix a student $i$, a set $S'\subseteq S$ with $|S'|\ge\delta n$, and an integer $L\ge1$. Consider the first $L$ schools in student $i$'s rank-ordered logit ranking. At any one of the first $L$ draws, conditional on the previously drawn schools, the probability that the next school lies in the as-yet-undrawn part of $S'$ is at least
	\[
	\frac{\delta n-L}{Cn}.
	\]
	For all sufficiently large $n$, since $L$ is fixed, this is at least
	\[
	q_{\delta,C}:=\frac{\delta}{2C}>0.
	\]
	Therefore the probability that none of student $i$'s first $L$ schools lies in $S'$ is at most
	\[
	(1-q_{\delta,C})^L.
	\]
	
	Now fix sets $I_0\subseteq I$ and $S'\subseteq S$ with $|I_0|,|S'|\ge\delta n$. Let $H_i$ be the indicator that student $i$ ranks at least one school in $S'$ among her top $L$ schools. The random variables $(H_i)_{i\in I_0}$ are independent across students, and
	\[
	\Pr(H_i=1)\ge 1-(1-q_{\delta,C})^L.
	\]
	Choose $L=L_\delta$ large enough that
	\[
	\delta\cdot \mathrm{KL}\left(\frac12\,\middle\|\,1-(1-q_{\delta,C})^L\right)>2\log 2+1,
	\]
	where $\mathrm{KL}(\cdot\|\cdot)$ denotes binary relative entropy. This is possible because $1-(1-q_{\delta,C})^L\to1$ as $L\to\infty$.
	
	By the Chernoff bound, for each fixed pair $(I_0,S')$,
	\[
	\Pr\left(\sum_{i\in I_0}H_i<\frac{|I_0|}{2}\right)
	\le
	\exp\left(-(2\log 2+1)n\right).
	\]
	There are at most $2^n$ choices of $I_0$ and at most $2^n$ choices of $S'$. Hence, by a union bound,
	\[
	\Pr\left(
	\exists I_0,S'\text{ with }|I_0|,|S'|\ge\delta n
	\text{ and }
	\sum_{i\in I_0}H_i<\frac{|I_0|}{2}
	\right)
	\le e^{-n}\to0.
	\]
	Thus, with high probability, simultaneously for every such pair $(I_0,S')$, at least $|I_0|/2$ students in $I_0$ rank some school in $S'$ among their top $L_\delta$ schools. Assumption~\ref{ass:entropy} follows with $\rho_\delta=1/4$.
\end{proof}

\begin{lemma}[Proposal volume]
	\label{lem:logit-proposals}
	For every fixed $C<\infty$, the bounded rank-ordered logit model with independently uniform priorities satisfies Assumption~\ref{ass:proposals}.
\end{lemma}

\begin{proof}
	In DA, the number of proposals made by a student equals the rank of her DA assignment. Hence it suffices to show that, for every fixed integer $L\ge1$ and every $\eta>0$, with high probability at most $\eta n$ students are assigned one of their top $L-1$ schools under DA.
	
	The bounded rank-ordered logit model with independently uniform priorities is a special case of the logit-based random matching model of \citet{ashlagi2023welfare}. In their notation, students' canonical score rows are proportional to $(\theta_j)_{j\in S}$, while schools' priority score rows are uniform. Since $1\le\theta_j\le C$, the canonical score ratios are uniformly bounded; hence their contiguity condition is satisfied.
	
	Their rank-distribution theorem and tail estimates imply that, with high probability, in every stable matching, all but an arbitrarily small fraction of agents have ranks bounded below by a sequence diverging with $n$.\footnote{See \citet[Theorem~7.4 and Corollary~C.8; see also Remark~C.9]{ashlagi2023welfare}. Their estimates imply that, for every $\eta>0$, with high probability at most $\eta n$ agents obtain ranks below a constant multiple of $(\log n)^{7/8}$ in any stable matching, up to the harmless side-specific fitness scaling. In the bounded public-score model considered here, these fitness values are within constant factors of $n$.}
	
	Since the DA outcome is stable, the same lower-rank bound applies to DA. For any fixed $L$, the diverging lower bound eventually exceeds $L$. Thus, with high probability, all but at most $\eta n$ students make at least $L$ proposals during DA. This is exactly Assumption~\ref{ass:proposals}.
\end{proof}

\begin{corollary}
	\label{cor:logit-giant}
	For every fixed $C<\infty$, the bounded rank-ordered logit model with independently uniform priorities satisfies the conclusion of Theorem~\ref{thm:correlated-giant}: for every fixed $\epsilon>0$, with high probability the DA envy digraph has a giant SCC containing at least $(1-\epsilon)n$ students.
\end{corollary}

\begin{proof}
	Immediate from Theorem~\ref{thm:correlated-giant}, Lemma~\ref{lem:logit-toplist}, and Lemma~\ref{lem:logit-proposals}.
\end{proof}

The bounded-ratio requirement is important. If the attractiveness ratio grows with $n$, a linear-sized set of low-attractiveness schools may be ranked below the first fixed number of positions by almost every student, and Assumption~\ref{ass:entropy} can fail. The additive-utility simulations below explore precisely this boundary: as the common component becomes dominant, the uniform expansion condition need not hold, even though the giant-SCC phenomenon may still appear over a wide range of parameters.

\subsubsection{Additive Utility Model: Evidence Beyond Sufficient Conditions}

Theorem~\ref{thm:correlated-giant} gives sufficient conditions under which the giant-SCC result extends to correlated preferences. However, these conditions are not designed to directly cover the additive utility model used by \cite{che2019efficiency}.
They parameterize correlation strength through a single variable $\alpha \in [0,1]$:
\begin{equation}
	u_{ij} = \alpha \cdot \theta_j + (1-\alpha) \cdot \epsilon_{ij}
\end{equation}
where:
\begin{itemize}
	\item $\theta_j \sim \text{Uniform}[0,1]$ is the common component representing school $j$'s objective quality/social desirability.
	\item $\epsilon_{ij} \sim \text{Uniform}[0,1]$ is the idiosyncratic component representing student $i$'s personal taste for school $j$.
	\item $\alpha \in [0,1]$ is the correlation parameter governing the strength of the common component.
	\item $\{\theta_j\}_{j=1}^n$ and $\{\epsilon_{ij}\}_{i=1,j=1}^{n,n}$ are mutually independent.
\end{itemize}

Student preferences are determined by ranking schools in descending order of utilities $u_{ij}$. When $\alpha = 0$, we recover the i.i.d. model; when $\alpha = 1$, preferences are perfectly correlated. Priorities remain uniformly random (as they are driven by idiosyncratic factors like distance and sibling attendance).

For high values of $\alpha$ close to 1, the uniform top-list expansion condition in Assumption~\ref{ass:entropy} is unlikely to hold. When the common component dominates, subsets of universally ``bad'' schools will be ranked low by nearly all students. Lacking a complete theoretical answer for this type of correlation structure, we provide an extensive simulation analysis documenting the emergence of a giant SCC for a large parameter family.

\paragraph{Simulation Design.} We simulate markets with $n \in \{100, 200, \ldots, 1000\}$ students and schools, generating 1,000 random instances for each market size and each correlation parameter $\alpha \in \{0.25, 0.5, 0.75\}$.

\paragraph{Findings.} Figures \ref{fig:alpha25}, \ref{fig:alpha5}, and \ref{fig:alpha75} show that:

\begin{enumerate}
	\item A unique giant SCC emerges for all moderate correlation values tested.
	\item The giant SCC contains an \emph{even larger} fraction of students than in the i.i.d. case ($\alpha = 0$), particularly for $\alpha = 0.5$ and $\alpha = 0.75$.
	\item Only for extreme values $\alpha > 0.9$ does the giant SCC become negligible, as preferences approach perfect correlation.
\end{enumerate}

\begin{figure}[h!]
	\centering
	\begin{subfigure}{0.32\textwidth}
		\includegraphics[width=\textwidth]{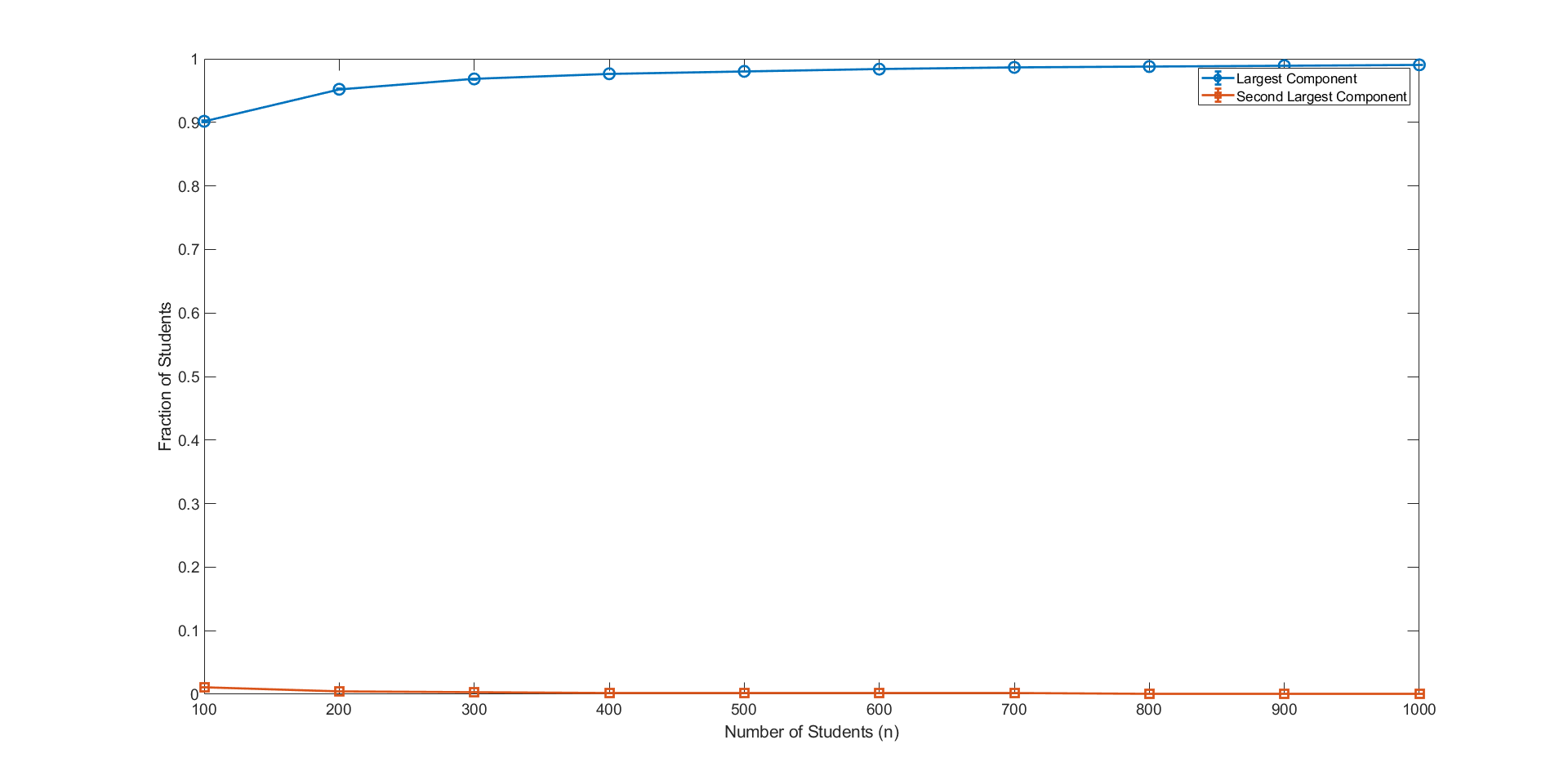}
		\caption{$\alpha=0.25$}
		\label{fig:alpha25}
	\end{subfigure}
	\hfill
	\begin{subfigure}{0.32\textwidth}
		\includegraphics[width=\textwidth]{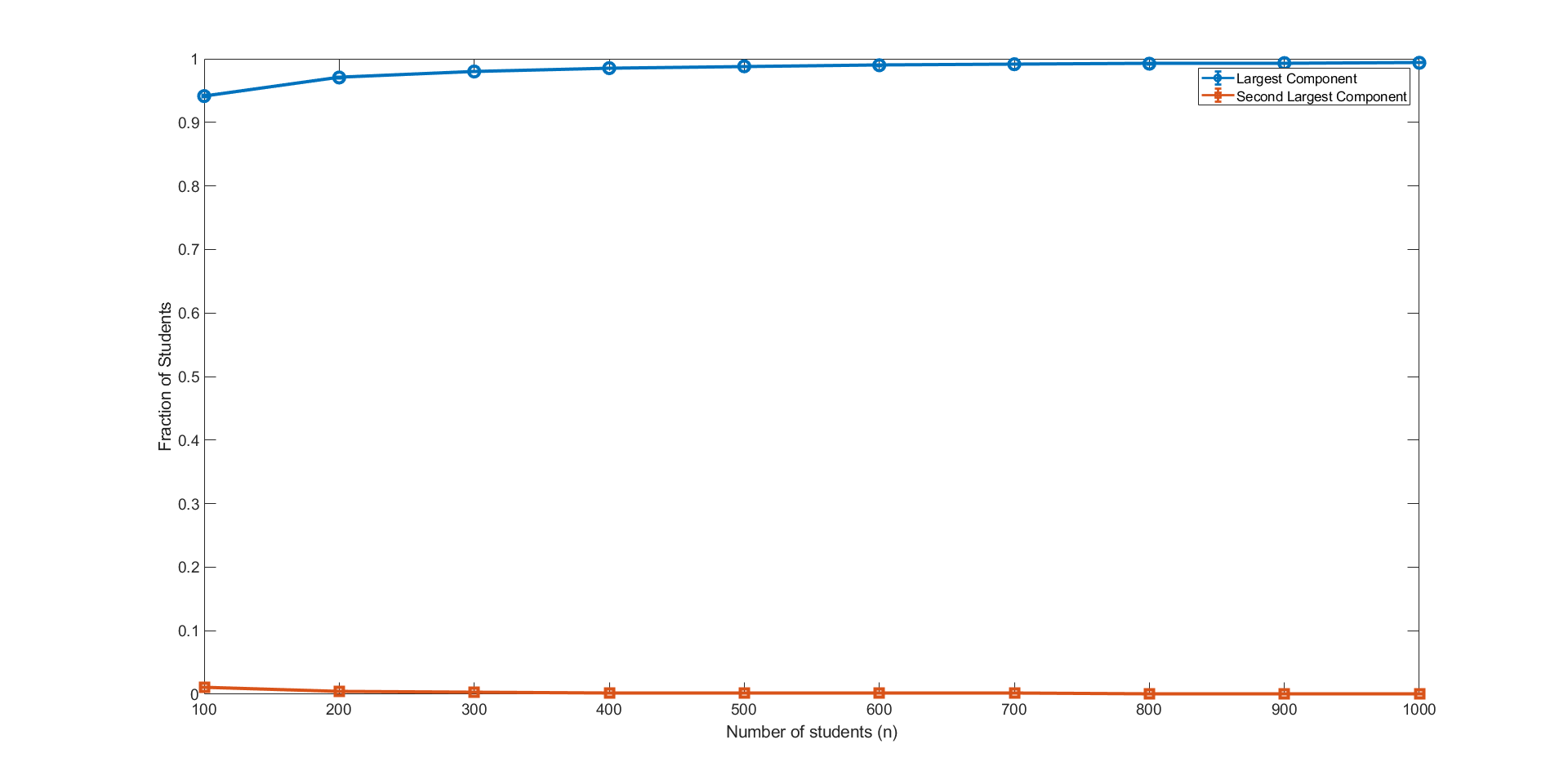}
		\caption{$\alpha=0.5$}
		\label{fig:alpha5}
	\end{subfigure}
	\hfill
	\begin{subfigure}{0.32\textwidth}
		\includegraphics[width=\textwidth]{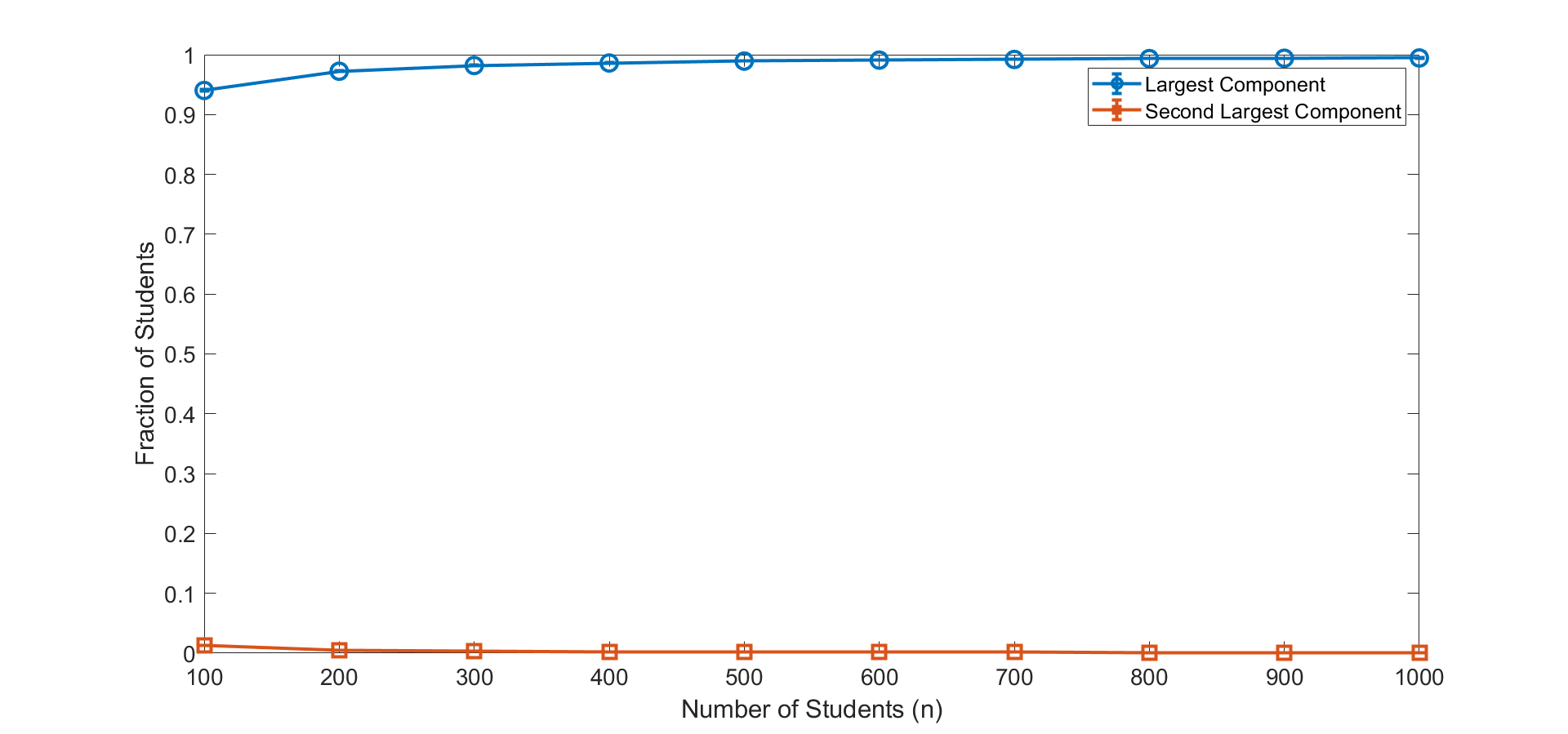}
		\caption{$\alpha=0.75$}
		\label{fig:alpha75}
	\end{subfigure}
	\caption{Fraction of nodes in largest (blue) and second-largest (red) SCCs under additive utility model for varying correlation parameters (1,000 replications per market size).}
	\label{fig:additive-all}
\end{figure}
\paragraph{Interpretation.} We conjecture that moderate correlation produces two competing effects: (1) it reduces the fraction of students obtaining their first choice (who are unimprovable by definition), while (2) it creates denser envy relationships among remaining students. For moderate $\alpha$, effect (2) dominates, actually increasing improvability relative to the i.i.d. baseline.

These simulations suggest the giant-SCC phenomenon can persist outside the sufficient conditions in Assumptions~\ref{ass:entropy} and~\ref{ass:proposals}. This indicates that the uniform top-list expansion and proposal-volume conditions are sufficient but not necessary for the emergence of the giant SCC.

In summary, the independence assumption in our main analysis does not represent a best-case scenario for DA. Rather, it may understate the prevalence of improvable students in empirically relevant settings with moderate correlation. The sufficient conditions we provide are conservative: they formalize when we can prove the result, but the phenomenon appears to hold more broadly.

\subsection{\new{Correlation in Priorities}}
\label{app:appendixthree}

\new{We now introduce correlation in schools' priorities to analyse whether the unique giant SCC remains. We introduce correlation using the additive model we used previously for preferences, where correlation enters via the parameter $\alpha \in [0,1]$ (we skip the tiered model in which students in higher tiers have universally higher priority at all schools, as this analysis yields as many giant SCCs as tiers, by an argument identical to the one used in Theorem 	\ref{thm:tiered} ).}

\new{We find, using simulations, that our results are qualitatively robust when correlation is mild or moderate: there remains a unique giant SCC, albeit smaller. When correlation becomes large $\alpha=0.75$, the size of the largest SCC collapses significantly for values of $n$ between 100 and 1,000, now only covering around 18\% of students (but still growing with $n$, Figure \ref{fig:additive-allpriorities}).}

\new{It is interesting that high (but not perfect) correlation in preferences implies a unique giant SCC larger than the iid baseline, yet high correlation in priorities has an opposite effect. An intuitive explanation for why this is the case is that correlation in preferences makes DA's envy digraph much denser, leading to significantly more cycles (alternatively, using Kesten's idea, it makes every student an interrupter that causes inefficiency with high probability). On the other hand, correlation in priorities does not increase the number of edges but changes its concentration: students with better priorities do not point to those with worse priorities, limiting the potential for cycles. }

\begin{figure}[h!]
	\centering
	\begin{subfigure}{0.32\textwidth}
		\includegraphics[width=\textwidth]{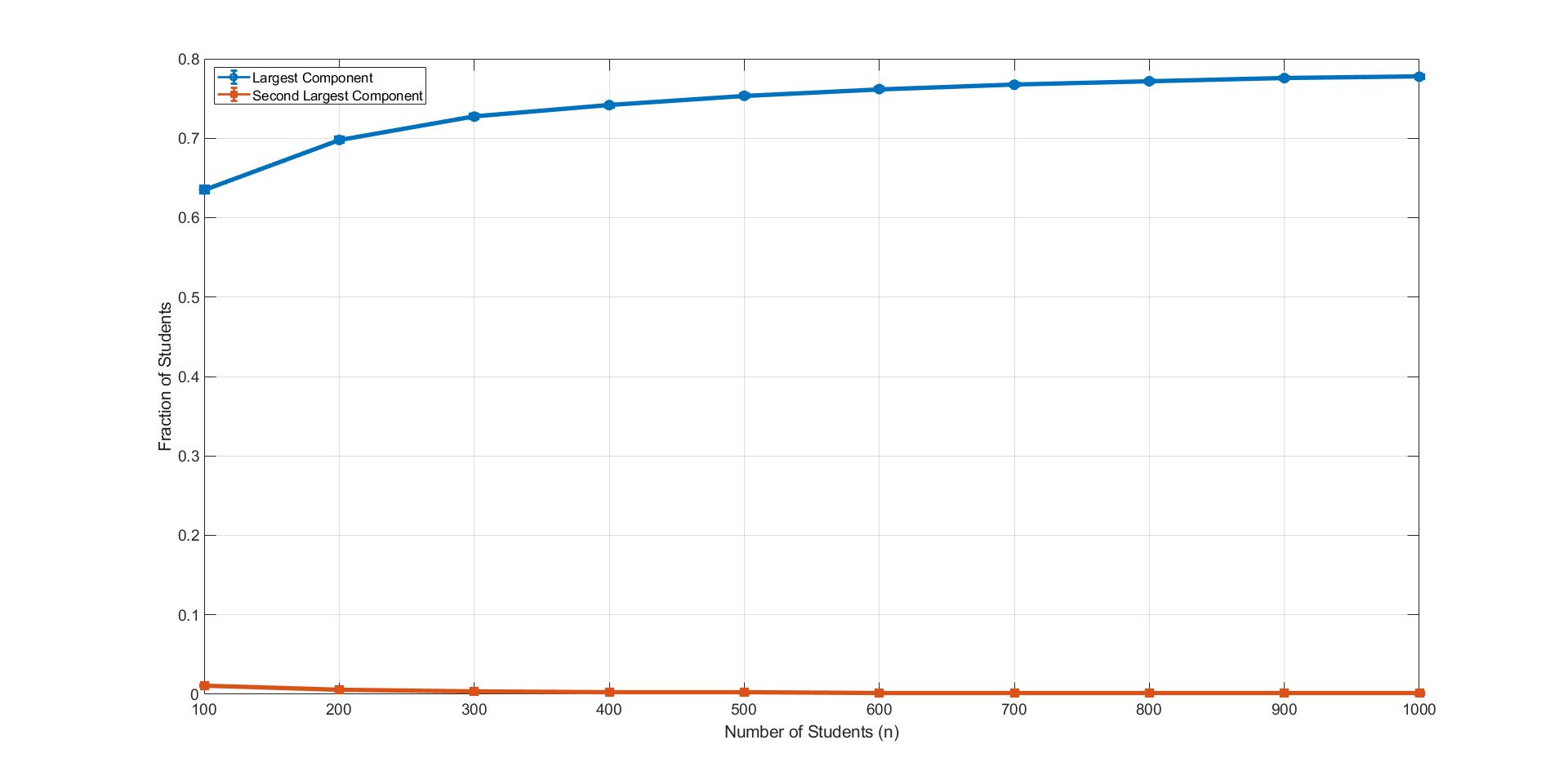}
		\caption{$\alpha=0.25$}
		\label{fig:alpha25priorities}
	\end{subfigure}
	\hfill
	\begin{subfigure}{0.32\textwidth}
		\includegraphics[width=\textwidth]{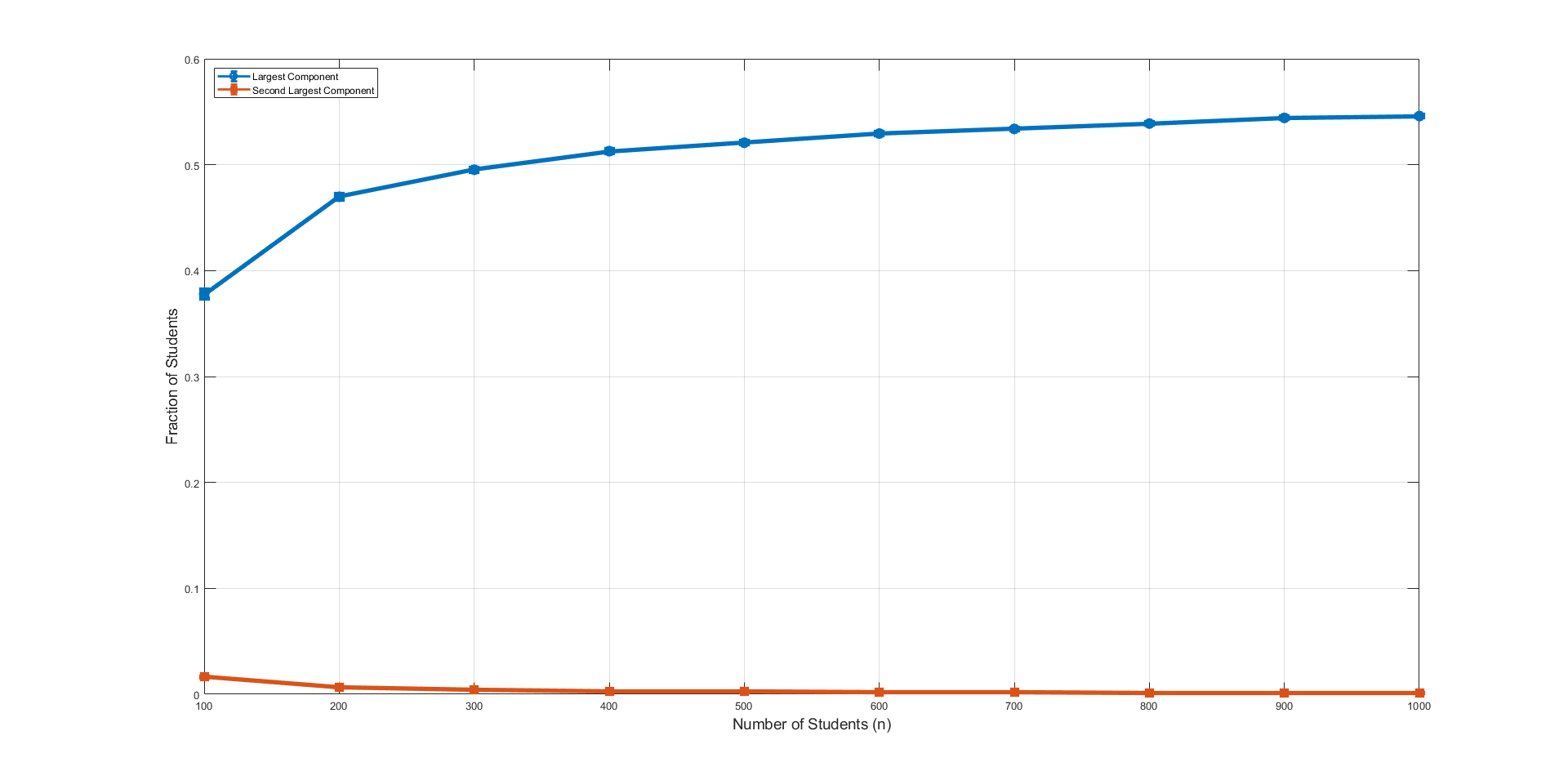}
		\caption{$\alpha=0.5$}
		\label{fig:alpha5priorities}
	\end{subfigure}
	\hfill
	\begin{subfigure}{0.32\textwidth}
		\includegraphics[width=\textwidth]{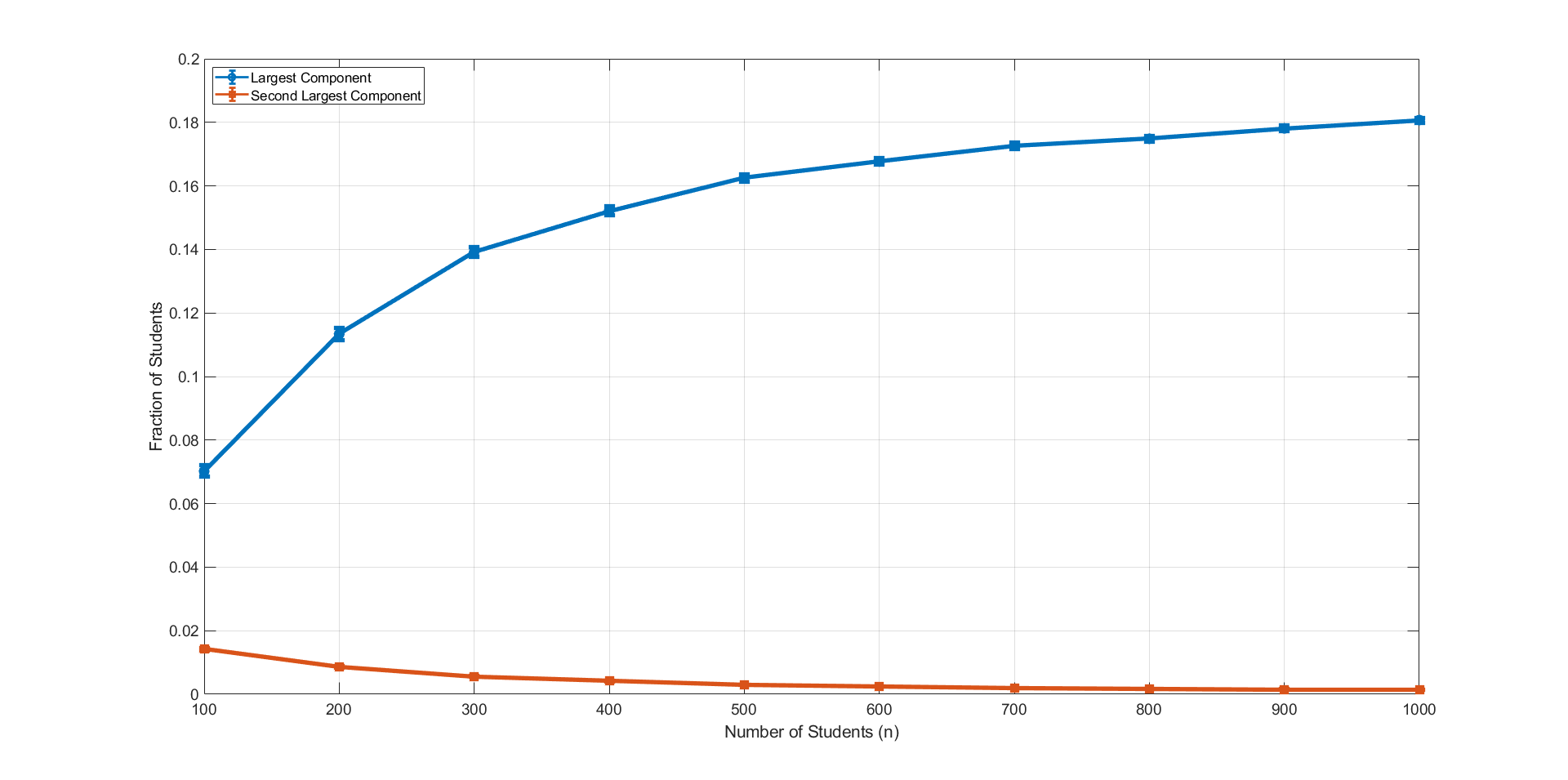}
		\caption{$\alpha=0.75$}
		\label{fig:alpha75priorities}
	\end{subfigure}
	\caption{\new{Fraction of nodes in largest (blue) and second-largest (red) SCCs under additive utility model for varying \emph{priorities'} correlation parameters (1,000 replications per market size).}}
	\label{fig:additive-allpriorities}
\end{figure}

\new{Since the effects of correlation in preferences and priorities differ, we conduct further simulations to analyze which effect dominates when both are present. The results appear in Figure \ref{fig:bothchanges}; $\alpha_i$ and $\alpha_s$ denote the correlation parameters for preferences and priorities, respectively. }

\begin{figure}[h!]
	\centering
	\begin{subfigure}{0.32\textwidth}
		\includegraphics[width=\textwidth]{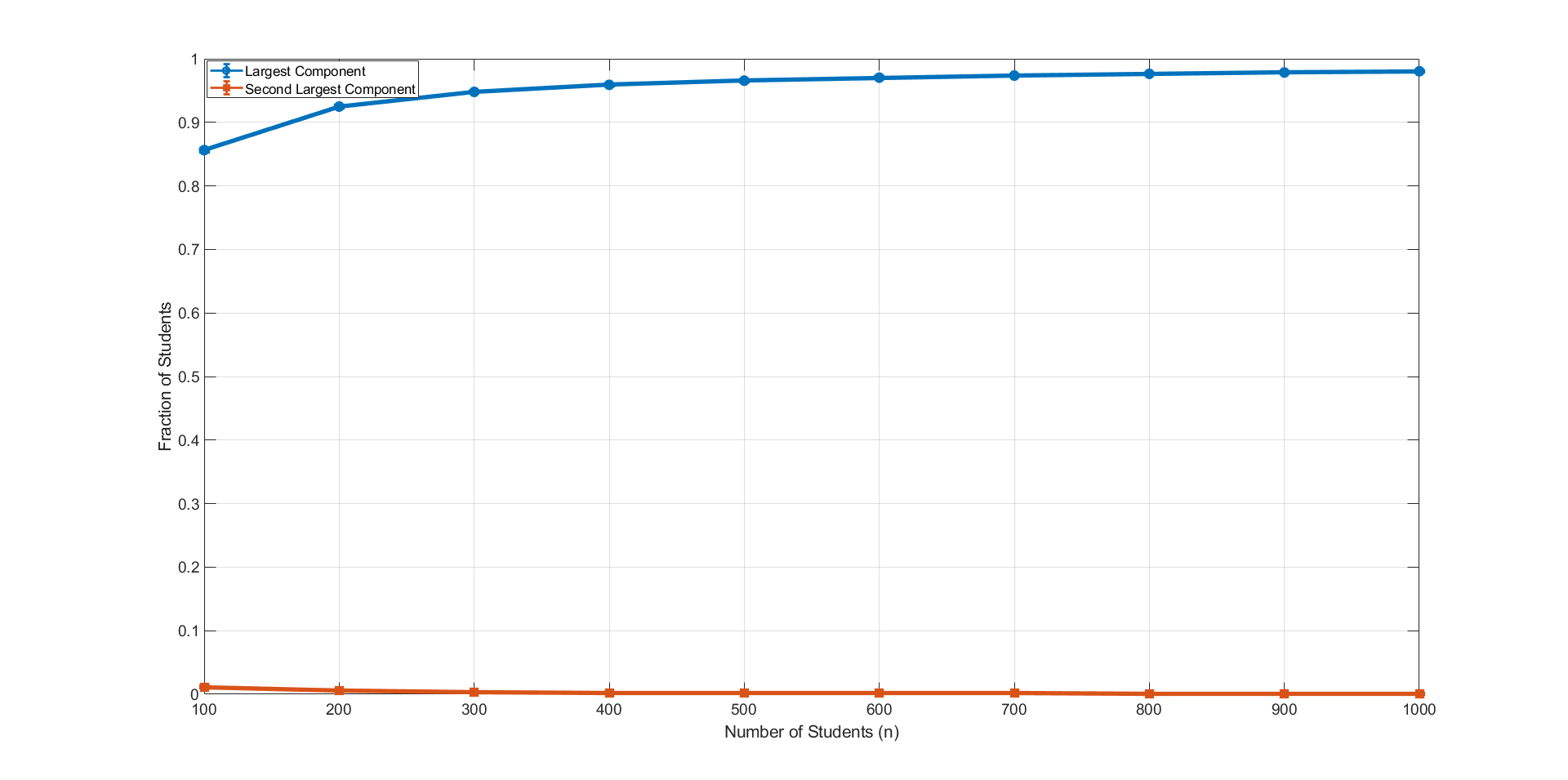}
		\caption{$\alpha_i=\alpha_s=0.25 $}
		\label{fig:alpha25priorities2}
	\end{subfigure}
	\hfill
	\begin{subfigure}{0.32\textwidth}
		\includegraphics[width=\textwidth]{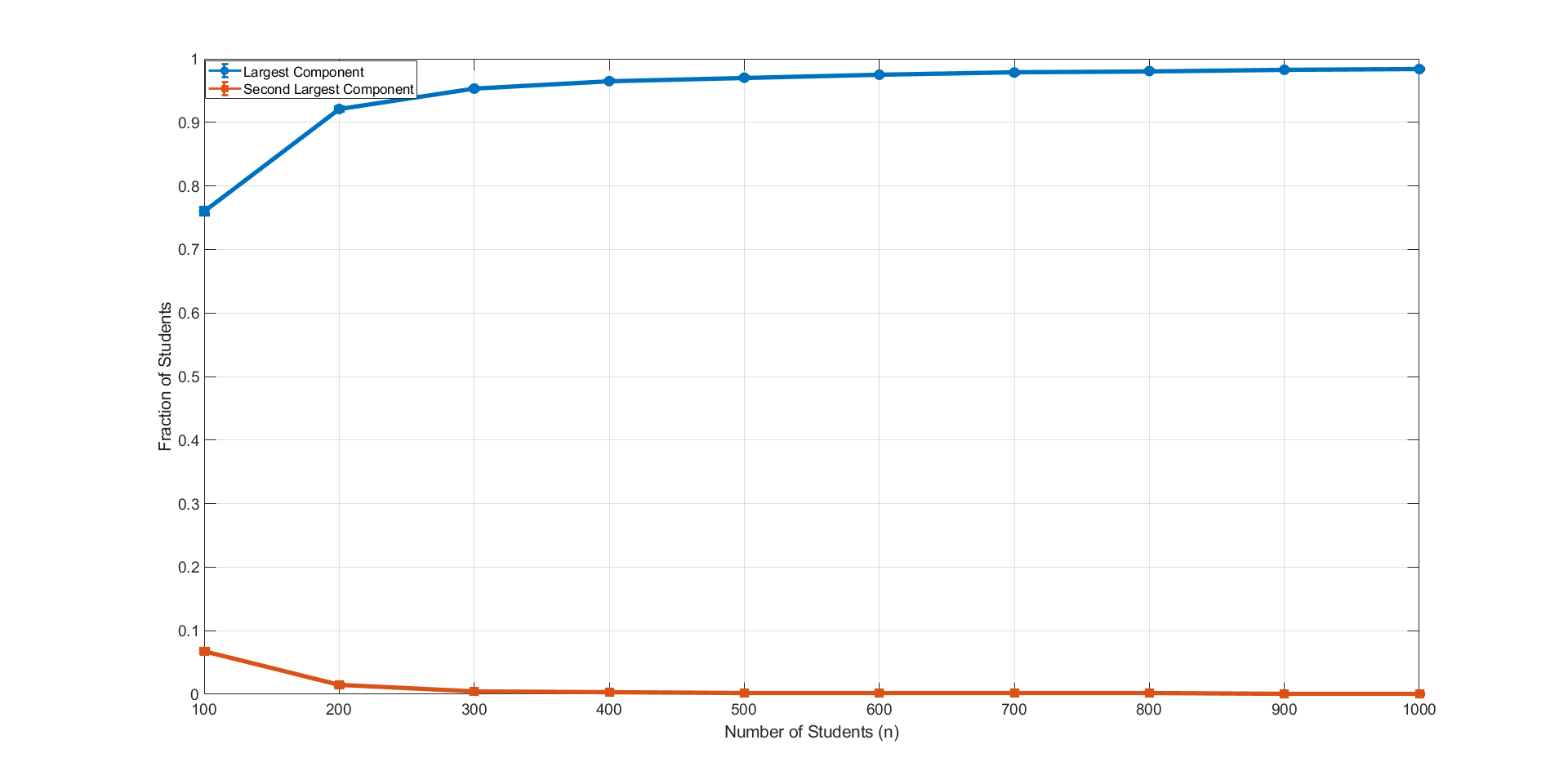}
		\caption{$\alpha_i=\alpha_s=0.5$}
		\label{fig:alpha5priorities2}
	\end{subfigure}
	\hfill
	\begin{subfigure}{0.32\textwidth}
		\includegraphics[width=\textwidth]{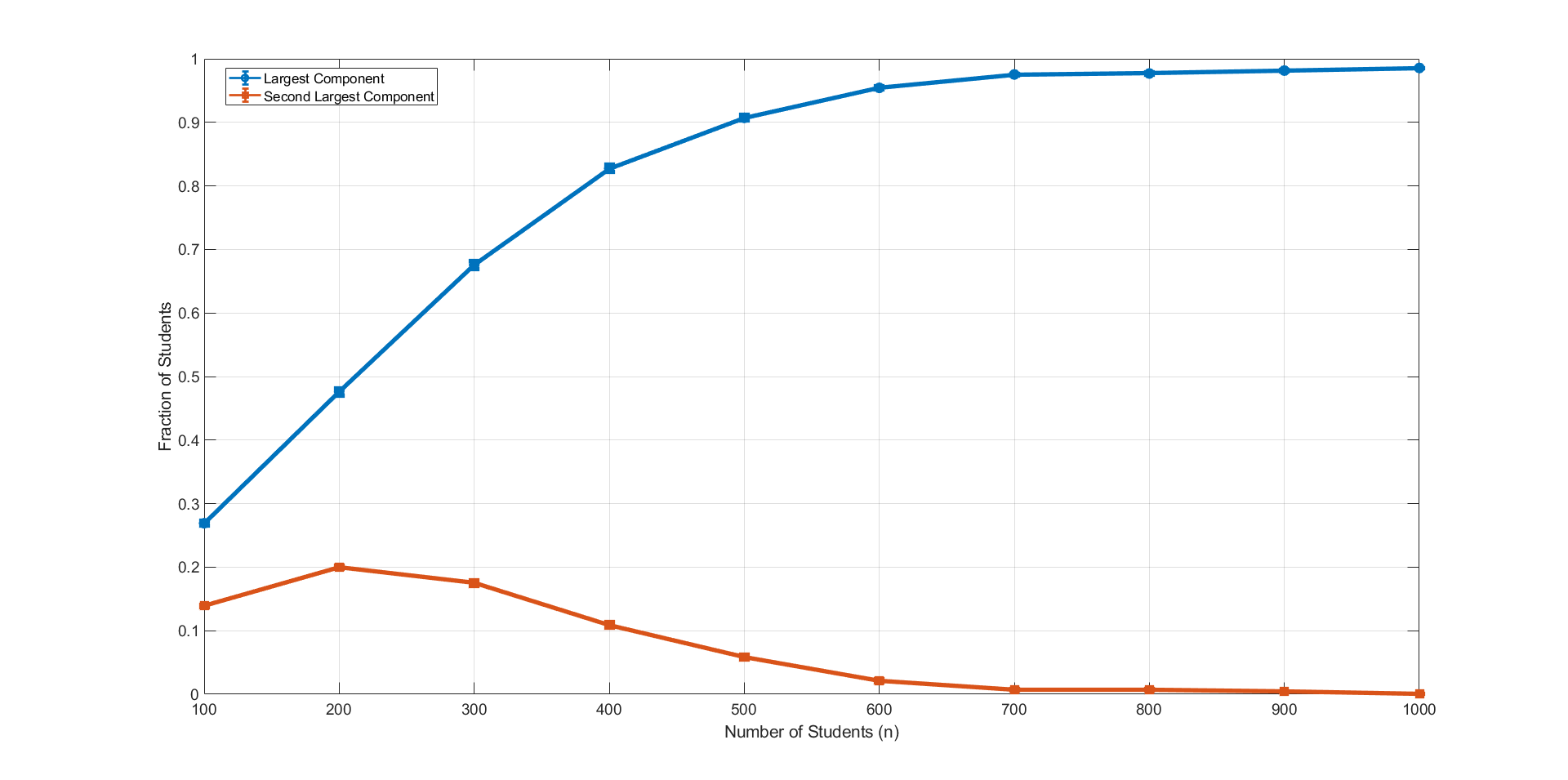}
		\caption{$\alpha_i=\alpha_s=0.75$}
		\label{fig:alpha75priorities2}
	\end{subfigure}
	\caption{\new{Fraction of nodes in largest (blue) and second-largest (red) SCCs under additive utility model for varying \emph{preferences' and priorities'} correlation parameters (1,000 replications per market size).}}
	\label{fig:bothchanges}
\end{figure}
\new{When introducing correlation on both sides, we find that the preference effect dominates decisively. Even at high correlation on both sides ($\alpha_i = \alpha_s = 0.75$), the giant SCC emerges and covers over 98\% of students by $n = 700$. For small markets under high bilateral correlation, two substantial components briefly coexist before merging into a single giant component (the same phase transition observed in the iid case, simply delayed to larger market sizes). When correlation is moderate on both sides ($\alpha_i = \alpha_s = 0.5$), the giant SCC reaches 99\% by $n = 1{,}000$, exceeding the iid baseline. These findings suggest that our main result is robust to realistic market structures  where either there is (not extreme) correlation in preferences, there is mild or moderate correlation in priorities, or there is (not extreme) correlation in both.}

\newpage\subsection{The Role of Priorities}
\label{app:role-priorities}

\new{In the main text, we assumed schools' priorities are independently and uniformly distributed, which forces a degree of heterogeneity in the way schools rank students. In this Appendix we expand on why the iid assumption implies that Ergin's acyclicity condition, which guarantees that DA is Pareto-efficient for every preference profile, is satisfied with super-exponentially small probability in our baseline model with $n$ students and schools with unit capacity.}

\new{An Ergin priority cycle is a tuple of two schools $a,b$ and three students $i,j,k$ such that
	\[
	i \triangleright_a j \triangleright_a k \triangleright_b i .
	\]
	By \citet[][Theorem 1]{ergin2002efficient}, DA is Pareto efficient for every preference profile if and only if the priority structure is acyclic. Alternatively, Ergin shows that the acyclicity condition holds if and only if, for every pair of schools $a,b$ and every student $i$, the rank of $i$ at $a$ differs from the rank of $i$ at $b$ by at most one.}

\new{In the iid benchmark, cyclic priorities are extremely likely.}

\begin{proposition}\label{prop:ergin-iid}
	\new{The probability that the priority structure is Ergin-acyclic converges to zero super-exponentially fast. In particular,
		\[
		\Pr(\text{priorities are Ergin-acyclic})
		=
		\exp\{-n^2\log n+O(n^2)\}.
		\]}
\end{proposition}

\begin{proof}
	Fix the priority order of one school and call it the reference order. We first note that pairwise acyclicity across all schools implies that the union of all adjacent pairs flipped by any school is itself a set of disjoint adjacent pairs in the reference order. Indeed, suppose two schools flip adjacent pairs that overlap in one student, say one flips $(x,y)$ and another flips $(y,z)$ in the reference order. Then the rank of $y$ differs by two across these two schools, violating Ergin's rank characterization. Hence overlapping flipped pairs cannot occur. Therefore there is a common set of disjoint adjacent pairs in the reference order such that each other school is obtained by flipping an arbitrary subset of these pairs.
	
	\new{This gives a unique representation of each acyclic profile: the common set is the union of all adjacent pairs flipped by at least one school, and for each pair we record the nonempty subset of schools that flip it. If this common set has size $k$, there are $\binom{n-k}{k}$ ways to choose it. For each chosen pair, there are $2^{n-1}-1$ nonempty subsets of the remaining $n-1$ schools that flip it. Hence, conditional on the first school's priority order,
		\[
		\Pr(\text{priorities are Ergin-acyclic})
		=
		\frac{1}{(n!)^{n-1}}
		\sum_{k=0}^{\lfloor n/2\rfloor}
		\binom{n-k}{k}
		\left(2^{n-1}-1\right)^k .
		\]}
	\new{The numerator has logarithm $O(n^2)$, whereas
		\[
		\log (n!)^{n-1}=n^2\log n+O(n^2)
		\]
		by Stirling's formula. Therefore
		\[
		\Pr(\text{priorities are Ergin-acyclic})
		=
		\exp\{-n^2\log n+O(n^2)\}.
		\]}
\end{proof}

\new{Proposition \ref{prop:ergin-iid} shows that the acyclic case is exceptional under iid priorities. We now strengthen this conclusion: the same is true throughout the additive correlation model from Appendix \ref{app:appendixthree}, in which each student $m$ has a common score $\theta_m \sim \mathrm{Uniform}[0,1]$ and each school $s$ ranks students by descending values of
	\[
	\pi_{s,m} \;=\; \alpha\,\theta_m \;+\; (1-\alpha)\,\epsilon_{s,m},
	\qquad \epsilon_{s,m} \stackrel{\mathrm{iid}}{\sim} \mathrm{Uniform}[0,1].
	\]
	The parameter $\alpha = 0$ recovers iid priorities; $\alpha = 1$ collapses to a master list, which is trivially Ergin-acyclic. The next result shows that every fixed correlation level strictly below this perfectly correlated boundary remains cyclic with high probability.}

\begin{proposition}\label{prop:ergin-correlated}
	\new{Under the additive priority model, for every fixed $\alpha \in [0,1)$,
		\[
		\Pr(\text{priorities are Ergin-acyclic}) \;\to\; 0
		\]
		exponentially fast as $n \to \infty$.}
\end{proposition}

\begin{proof}
	\new{The case $\alpha = 0$ is covered by Proposition \ref{prop:ergin-iid}. Fix $\alpha \in (0,1)$ and a tuple $(a,b;i,j,k)$ of two distinct schools and three distinct students. We show that the probability of an Ergin cycle on this tuple is bounded below by a positive constant $p(\alpha)$ independent of $n$.}
	
	\new{Set $\delta = \min\{1/4,\,(1-\alpha)/(4\alpha)\} > 0$ and define
		\[
		E_\delta \;=\; \big\{\, |\theta_x - \theta_y| \leq \delta \text{ for all } x,y \in \{i,j,k\}\,\big\}.
		\]
		Since $\theta_i,\theta_j,\theta_k$ are iid Uniform$[0,1]$, $\Pr(E_\delta) \geq c_1\,\delta^2$ for some absolute constant $c_1 > 0$.}
	
	\new{Conditional on $E_\delta$ and on the realised values of $(\theta_i,\theta_j,\theta_k)$, the Ergin cycle condition $i \triangleright_a j \triangleright_a k \triangleright_b i$ is equivalent to the system
		\[
		\begin{aligned}
			(1-\alpha)(\epsilon_{a,i} - \epsilon_{a,j}) &> \alpha(\theta_j - \theta_i),\\
			(1-\alpha)(\epsilon_{a,j} - \epsilon_{a,k}) &> \alpha(\theta_k - \theta_j),\\
			(1-\alpha)(\epsilon_{b,k} - \epsilon_{b,i}) &> \alpha(\theta_i - \theta_k),
		\end{aligned}
		\]
		in the independent shocks $\{\epsilon_{a,i},\epsilon_{a,j},\epsilon_{a,k},\epsilon_{b,i},\epsilon_{b,k}\}$. Dividing through by $(1-\alpha)$, each right-hand side is bounded in absolute value by $\alpha\delta/(1-\alpha) \leq 1/4$. The first two inequalities involve only the $\epsilon$ shocks at school $a$ and are independent of the third, which involves only the shocks at school $b$. Direct integration shows that, for any shifts of absolute value at most $1/4$, the probability of this system being satisfied is bounded below by a positive absolute constant $c_2 > 0$ (the worst case is shifts equal to $1/4$, which gives probability exactly $1/48 \cdot 9/32$).}
	
	\new{Therefore
		\[
		p(\alpha) \;:=\; \Pr\!\big(\,i \triangleright_a j \triangleright_a k \triangleright_b i\,\big) \;\geq\; c_1 c_2\,\delta(\alpha)^2 \;>\; 0.
		\]
		Take $\lfloor n/3 \rfloor$ disjoint blocks of two schools and three students each. The $\theta$ and $\epsilon$ variables governing different blocks are independent, so the cycle events are independent across blocks, and
		\[
		\Pr(\text{priorities are Ergin-acyclic}) \;\leq\; \big(1 - p(\alpha)\big)^{\lfloor n/3 \rfloor} \;\to\; 0
		\]
		exponentially fast as $n \to \infty$.}
\end{proof}

\new{The acyclic-priority case is therefore a singular boundary of the additive parameter space: every fixed $\alpha \in [0,1)$ produces Ergin cycles with high probability, and only the perfectly correlated case $\alpha = 1$ yields an acyclic priority structure. This complements the simulation evidence in Appendix \ref{app:appendixthree}: as the fixed value of $\alpha$ is chosen closer to 1, the exponential rate $p(\alpha)$ shrinks (since $\delta(\alpha) \to 0$), mirroring the gradual contraction of the giant SCC observed there, but the rate stays strictly positive for each fixed $\alpha<1$.}
	\end{document}